\documentclass[a4paper]{article}

\usepackage[T1]{fontenc}
\usepackage[utf8]{inputenc}
\usepackage[english]{babel}																					
\usepackage{amsmath,amsfonts,amssymb,amsthm,mathtools,mathabx,amscd,	
		amsgen,amstext,amsbsy,amsopn,amsxtra}	
\usepackage{bm,bbm}										 													
\usepackage{graphicx}																							
\usepackage[hidelinks]{hyperref}																			
\usepackage[shortlabels]{enumitem}																		
\usepackage{latexsym}																							
\usepackage{mathrsfs}																							
\usepackage{xspace}																							
\usepackage{color}																								
\usepackage[sort&compress,capitalise,nameinlink]{cleveref}									
\usepackage[font=footnotesize]{caption}																
\usepackage{tikz}																									
	\usetikzlibrary{angles,quotes,patterns}
\usepackage{geometry}																							
\usepackage{amsmath}																							

\usepackage{authblk}

\usepackage{ulem}																									

\DeclareMathAlphabet{\mathpzc}{OT1}{pzc}{m}{it}													

\newtheorem{teo}{Theorem}[section]
\newtheorem{lem}{Lemma}[section]
\newtheorem{pro}{Proposition}[section]

\newtheorem{asum}{Assumption}

\theoremstyle{definition}
\newtheorem{defi}{Definition}[section]
\newtheorem{remark}[teo]{Remark} 

\numberwithin{equation}{section}

\newcommand{\eps}{\varepsilon}

\renewcommand{\Im}{\mathrm{Im}}

\newcommand{\pgpt}[1]{{\varphi^{\mathrm{GP}}_{#1}}}
\newcommand{\pht}[1]{{\varphi^{\mathrm{H}}_{#1}}}

\DeclareMathOperator{\tr}{tr}

\begin{document}

\title{On Bose-Einstein condensates in the Thomas-Fermi regime}

\author[1]{Daniele Dimonte}
\affil[1]{Universit\"at Basel, Spiegelgasse 1, CH-4051 Basel, Switzerland.}

\author[2]{Emanuela L. Giacomelli}
\affil[2]{LMU M\"unich, Department of Mathematics, Theresienstr. 39, 80333 M\"unchen, Germany}

\maketitle 

\begin{abstract}

	We study a system of $ N $ trapped bosons in the Thomas-Fermi regime with an interacting pair potential of the form $ g_N N^{3\beta-1} V(N^\beta x) $, for some $ \beta\in(0,1/3) $ and $ g_N $ diverging as $ N \to \infty $. We prove that there is complete Bose-Einstein condensation at the level of the ground state and, furthermore, that, if $ \beta\in (0,1/6) $, condensation is preserved by the time evolution.

\end{abstract}

\tableofcontents

\section{Introduction and main results}\label{sec:intro}

In recent years, the analysis of quantum systems of interacting particles, in particular bosons, in suitable scaling regimes (mean-field, Gross-Pitaevskii, thermodynamic, etc.) has represented a very flourishing line of research in mathematical physics. Starting from the pioneering works \cite{GV, H, Sp}, several results have been obtained along this research line about the ground state behavior of interacting bosons and Bose-Einstein condensation (BEC) \cite{LS1, LSSY, LSY, LY} in both the mean-field (see \cite{GS, LNR1, LNR2, LNSS, Pz2, S} and references therein), in the Gross-Pitaevskii (GP) regimes (see \cite{BBCS2, BBCS3, Ha, NNRT, NRS, NT} and references therein) and in the thermodynamic limit (see \cite{BaCS,FS,YY} and references therein), as well as about the many-body dynamics in the mean-field (see \cite{ESY071, KP, P1} and references therein) and in the GP (see \cite{BCS1, BS19,  ESY10, P2} and references therein) settings. As we are going to see in more details later, in both these regimes, the effective behavior of the many-body system is shown to be suitably approximated studying an effective one-particle nonlinear problem.

In the typical setting, a large number $ N $ of bosonic particles is confined in a box or, more generally, by a trapping potential, and the bottom of the spectrum of the corresponding Hamiltonian, or its unitary evolution, is investigated as $ N \to \infty $. The many-body Schr\"{o}dinger operator has the form
\begin{equation}\label{eq: H_N}
	H_N^{\mathrm{trap}} := \sum_{j=1}^N \left(-\Delta_j +V_{\mathrm{ext}}(x_j)\right) + g_N N^{3\beta-1}\sum_{1\leq j < k \leq N} v(N^\beta (x_j-x_k)),
\end{equation}
and it is supposed to act on $L^2_{\mathrm{s}}(\mathbb{R}^{3N})$, which is the subspace of $L^2(\mathbb{R}^{3N})$ consisting of all functions which are symmetric with respect to permutations of the N particles in 3 dimensions. Here, $V_{\mathrm{ext}}$ is the trapping potential, which for simplicity is assumed to be homogeneous of order\footnote{The condition $ s \geq 2 $ is actually assumed only to simplify the presentation of the result, which holds true with different (more involved) estimates of the error terms also for $ s \geq  0 $.} $ s  $ (see also \cref{rem: asympt homo}), i.e.,
\begin{equation}
	\label{eq: ext potential}
	V_{\mathrm{ext}}(x) = \lambda |x|^s,		\qquad	\mbox{with } s \geq 2 \mbox{ and } \lambda > 0,
\end{equation}
while $N^{3\beta} v(N^\beta x)$ (with $\beta >0$) is the pair interaction potential, which is suitably rescaled by the prefactor $ g_N/N $. In the following we use the notation $v_N(x):=  N^{3\beta} v(N^\beta x)$.
The parameter $ \beta > 0 $ varies in $ [0,1]$ and, with $ g_N = g = \mathrm{const.} $, identifies different scaling regimes: mean-field for $ \beta = 0 $, mean-field and zero-range for $ \beta \in (0,1) $ and GP for $ \beta = 1 $. In fact, $ \beta $ is also a measure of the {\it diluteness} of the system: indeed, denoting by $ \bar{\rho} $ the average density, we have $ \bar{\rho} \simeq N \left\| \rho^{\mathrm{GP}}_N \right\|_{\infty} $ (see \cite{BCPY}), where $  \rho^{\mathrm{GP}}_N $ is the density associated to the condensate wave function.

In this paper we will investigate the Thomas-Fermi (TF) limit, where $ g_N\to\infty $. In the TF regime, one thus has (see \cref{rem: GP minimizer}) 
\[
	\bar{\rho} R_{N}^3 \propto \frac{N^{1-3\beta}}{g_N^{\frac{3}{s+3}}},
\]
where $ R_N \sim N^{-\beta} $ is the effective range of the interaction potential and $ g_N^{\frac{1}{s+3}} $ is the size of the TF radius, i.e., the characteristic length scale of the condensate (see Remark \ref{rem: GP minimizer}). Hence, since $ g_N \to +\infty $ as $ N \to  \infty $, the above quantity certainly tends to $ 0 $ for $ \beta > 1/3 $, so that the gas is dilute and only binary scattering events are relevant; on the opposite, for $ \beta < 1/3 $, the quantity may diverge, depending on the how fast $ g_N $ tends to $ + \infty $, so that any particle in the gas can interact with all the others, as in the typical mean-field picture. 

The final outcome in terms of effective description is physically different but it is remarkable that, for any $ \beta \in (0,1] $, a one-particle Schr\"{o}dinger operator with cubic nonlinearity  is recovered as $ N \to \infty $, giving rise to the GP functional
\begin{equation}
	\label{eq: GP functional}
	\mathcal{E}^{\mathrm{trap}}_{\mathrm{GP}}(u)  := \langle u, h u\rangle +4\pi g \int_{\mathbb{R}^3}dx\,|u(x)|^4
\end{equation}
in the stationary setting, where $h$ is the one-body operator 
\begin{equation}
	h := -\Delta + V_{\mathrm{ext}},
\end{equation}
and to the GP time-dependent equation
\begin{equation}
	i \partial_t \psi_t = (- \Delta  + 8 \pi g |\psi_t|^2) \psi_t,
\end{equation}
in the dynamical one: note the absence of the trapping potential in the above evolution equation, which is usually switched-off after the occurrence of BEC to observe the gas dynamics. Only the coefficient $ g $ of such a nonlinear term changes for different values of $ \beta $: for $ \beta \in (0,1) $ the coefficient is proportional to the integral of the interaction potential (which has then to be assumed finite), while, for $ \beta = 1 $, its scattering length appears. In both cases the coefficient $g$ is assumed to be constant as $N\rightarrow \infty$.

The physical motivation behind the choice of letting $ g_N $ grow with $ N $ is that in physical heuristics, the effective coefficient $ g $ is very often assumed to be large and the asymptotics $ g \to + \infty $ considered, which goes under the name of {\it Thomas-Fermi regime}, because of a similarity at the level of the one-particle density with the Thomas-Fermi theory of electrons. The main reason behind such a choice is that in real experiments $ g \propto a N $, $ a $ being the measured scattering length of the interaction potential, and therefore, even in presence of few thousands of atoms, the quantity might be quite large compared to the other characteristic lengths of the system. 

Another important motivation to let $g = g_N$ diverge is to take into account other settings and, in particular, the thermodynamic limit, i.e., the regime in which $ N, L \to  \infty $, but $ \rho : = N/L^3 $ is kept fixed: after rescaling the lengths this limit can be recast in the form above and recovered for $\beta = 1/3$, $g_N = N^{2/3}$. Finally, it is remarkable that a quite consistent literature is available about the ground state behavior of the GP functional in the asymptotic regime $ g \to + \infty $, in particular in presence of rotation (see \cite{CPRY2,CRY} and references therein), when vortices might appear as a manifestation of the superfluidity of the system. 

In this very same setting, the effective dynamics of vortices can also be studied \cite{JS} and shown to be driven by the trapping potential: if at initial time a trapped bosonic system is set in a condensate state with finitely many vortices (e.g., by putting it under sufficiently fast rotation) and then it is let to evolve (switching off the rotation), it is possible to follow the trajectories of vortex points and prove that they are determined by an ODE depending only on the external trap. We are going to comment further on this result below but we point out here that two key ingredients are needed in its proof: the system is assumed to be in the TF regime and the effective evolution equation for the condensate is given by the time-dependent GP equation. Here, we are precisely concerned with proving that such assumptions make sense. More precisely, we show that 
\begin{enumerate}
	\item in the TF regime (for $\beta < 1/3$, $g_N \ll N$) Bose-Einstein condensation takes place at the level of the ground state of a system of interacting bosons;
	\item in the same TF limit, the effective time-evolution of a state with sufficiently low energy is governed by the time-dependent GP equation.
\end{enumerate}  

A similar regime has indeed been studied so far at the many-body level only in \cite{ABS, BCS,BCPY}. However, in \cite{BCPY} the focus is different: the object of the study is the interplay between the TF regime and the presence of a rapid rotation and its effect to the ground state energy is investigated to leading order. We comment later on the similarities between \cite{ABS, BCS} and our work.

Let us now provide some more details about the setting we are going to consider. The main object under investigation is the many-body Hamiltonian \eqref{eq: H_N} and its time evolution, acting on bosonic (i.e., symmetric) wave functions in $L^2_{\mathrm{s}}(\mathbb R^{3N})$. We restrict our analysis to interaction potentials satisfying the assumptions below.

\begin{asum}\label{asump: 1}
 The potential $v\neq 0$ is of positive type, i.e., $\widehat{v}\geq 0$ and spherically symmetric. Furthermore, we assume that $|x|v\in L^1(\mathbb{R}^3)$, $v\in L^2(\mathbb{R}^2)$. 
\end{asum}

Let us comment further about the characteristic length scales of the system. The most important one is indeed the range of the interaction potential $ \propto N^{-\beta} $ and, if $ g_N $ is not too large, we claim that it is the only relevant one to be compared with the mean free path of the particles $ \propto N^{-1/3} $. This yields the threshold at $ \beta = 1/3 $ mentioned above. Another important length scale is provided by the {\it scattering length} associated to $ v_N $: for a potential $v$ the scattering length is defined through the solution of the zero energy scattering equation, i.e., 
\begin{equation}
	\left(-\Delta + \tfrac{1}{2}v\right) f = 0,
\end{equation}
with the boundary condition $f(x)\rightarrow 1$ as $|x|\rightarrow \infty$. The scattering length $a$ is then given by 
\begin{equation}
	8\pi a = \int_{\mathbb{R}^3} dx\, v(x)f(x).
\end{equation}
If we denote by $a_N$ the scattering length of the potential $g_N N^{3\beta -1} v(N^\beta x) = g_N v_N(x) N^{-1}$ then one can go through the Born approximation of $ a_N $ to write
\begin{equation}
	a_N = \frac{g_N }{ 8\pi N}\int_{\mathbb{R}^3}\, dx\, v_N(x) f_N(x) \sim  \frac{g_N }{ 8\pi N} \int_{\mathbb{R}^3}\, dx\, v(x),
\end{equation}
where $f_N$ is the solution of the related scattering equation, i.e., $-\Delta f_N +(g_N/2)v_N f_N = 0$. In order for the expansion to make sense, one has to require that $ v $ is integrable and
\begin{equation}
	g_N \ll N^{1-\beta},
\end{equation}
which we always do, but at the same this implies that 
\begin{equation}
	a_N \ll R_{N} \sim N^{-\beta},
\end{equation}
for any $ \beta \in [0,1) $. In other words, $ a_N $ is always much smaller than the range of the potential in the regime we are going to consider and plays no role in the picture. A transition is expected to occur when $ g_N \sim N^{1-\beta} $, for $ \beta \geq 1/3 $, which is precisely the case considered\footnote{In the notation of \cite{BCS}, $ \beta = 1 - \kappa $ and $ g_N = N^{\kappa} $, so that $ g_N = N^{1 - \beta} $.} in \cite{BCS}. The last characteristic length of the problem is the {\it healing length} $ \xi_N $, which is given by (see \cite[Sect. 4.1]{FF})
\begin{equation}	
	\xi_N = \frac{1}{\sqrt{N \left\| \rho^{\mathrm{GP}}_N \right\|_{\infty}}} \sim N^{-\frac{1}{2}} g_N^{-\frac{3}{2(s+3)}}.
\end{equation} 
Hence, one can readily check that under the conditions of next \cref{thm: E0}
\begin{equation}	
	a_N \ll \xi_N \ll R_{N}.
\end{equation}

We denote by $E_0(N)$ the ground state energy of the system, i.e.,
\begin{equation}
	E_0(N) := \inf \sigma(H_N^{\mathrm{trap}}) \equiv \inf \{ \langle \psi, H_N^{\mathrm{trap}} \psi \rangle \, \vert\, \psi \in L^2_s(\mathbb{R}^{3N}), \, \|\psi\|_{L^2_s(\mathbb{R}^{3N})} = 1 \}.
\end{equation}
Moreover, $\Psi_0\in L^2_s(\mathbb{R}^{3N})$ is the unique (up to an overall phase) ground state of $H_N$ and we denote by $\gamma_{\Psi_0}$ the associated reduced density matrix, i.e.,
\begin{equation}\label{eq: red dens}
	\gamma_{\Psi_0} := \mathrm{Tr}_{2,\cdots,N}|\psi_0\rangle \langle \psi_0|,
\end{equation}
where with $ \mathrm{Tr}_{2,\cdots,N} $ we denoted the partial trace on all particles but one.

It is well known that in the mean-field regime ($ \beta = 0 $, $ g_N = g $ constant), the ground state energy per particle is well approximated, in the limit $N\rightarrow\infty$, by the infimum of the Hartree functional, which approximately coincides with the restriction of the quadratic form associated to $H_N$ to functions of the form $\psi= u^{\otimes N}$, with $u\in L^2(\mathbb{R}^3)$, $\|u\|_{L^2(\mathbb{R}^3)} = 1$, i.e.,
\begin{equation}
	\mathcal{E}^{\mathrm{trap}}_{\mathrm{H}}(u) := \langle u, h u\rangle + \tfrac{1 }{ 2}\langle u\otimes u, v u\otimes u\rangle.
\end{equation}	
More precisely, in \cite{LNR1}, it is shown that
\begin{equation}
	\lim_{N\rightarrow \infty}\frac{E_0(N) }{ N}
	=\inf\left\{\mathcal{E}^{\mathrm{trap}}_{\mathrm{H}}(u)|\ \|u\|_{L^2\left(\mathbb R^3\right)} = 1\right\}	=: E^{\mathrm{H}}.
\end{equation}	
For trapped bosons in the mean field regime, it has been proven that the system exhibits BEC and its excitation spectrum can be approximated via the well-known Bogoliubov approximation \cite{GS}, with corrections that can be explicitly identified \cite{BSS3,BPPS}. Analogous results are proven for a system of bosons in a box in the same regime \cite{DN, LNSS, S}.

Similarly, it is also known that in the GP regime ($ \beta  = 1 $, $ g_N = g $ constant), the ground state energy per particle can be approximated, in the limit $N\rightarrow \infty$, by the infimum of the GP effective energy \eqref{eq: GP functional}.
Indeed, in \cite{NRS} it is proven that, in the limit $N\rightarrow \infty$, it holds (recall \eqref{eq: GP functional})
\begin{equation}
	\lim_{N\rightarrow \infty}\frac{E_0(N) }{ N}
	=\inf\left\{\mathcal{E}^{\mathrm{trap}}_{\mathrm{GP}}(u)|\ \|u\|_{L^2\left(\mathbb R^3\right)} = 1\right\}
	=: E^{\mathrm{trap}}_{\mathrm{GP}}.
\end{equation}
The occurrence of BEC in this case has been proven \cite{BSS, NNRT} as well as the validity of the Bogoliubov approximation \cite{BSS1, NT}. Similar results both on BEC and the excitation spectrum are available in (see \cite{ BBCS3, Ha, H} and references therein) for $N$ interacting bosons in a box.

Concerning the intermediate regimes, we have to distinguish between two different cases. The first interesting regime is studied in \cite{BBCS1} and it is characterized by taking $ \beta \in (0,1) $ and $ g_N = g >0$ small enough. In \cite{BBCS1} the validity of the Bogoliubov theory is proven in the aforementioned regime, studying both the ground state energy and the excitation spectrum of the system. More recently another regime has been investigated, which can be thought as an interpolation between the GP regime and the thermodynamic limit and it is described by $\beta\in (0,1)$ small enough and $g_N = N^{1-\beta}$. In \cite{ABS,BCS} it is proven that, in such regime, the low-energy states exhibit BEC and the low-energy excitation spectrum is studied. Note that, according to the discussion above on the characteristic length scales of the problem, this choice corresponds to a scattering length of the interaction of the same order of its effective range.

In the present work, on the other hand, we are mostly interested in the time-evolution of condensates in the TF regime ($ g_N \gg 1 $). For the sake of completeness, we also show that BEC occurs in the same setting, at least at the level of the ground state or of any approximate ground state, as it can be easily obtained by a suitable adaptation of the arguments in \cite{GS,S} and with a much simpler proof  (thanks to stronger assumption on $ \beta $ and $g_N$) than the ones in \cite{ABS, BCS}. More precisely, we prove that the ground state energy of \eqref{eq: H_N}, as long as $\beta < 1/3$, exhibits BEC in the (unique) minimizer of the GP functional, which in our setting reads as
\begin{equation}\label{eq: EGP gN}
	\mathcal{E}^{\mathrm{trap}}_{\mathrm{GP},N}(\psi) := \langle \psi, h\psi \rangle + \frac{g_N }{ 2}(\smallint v)\int_{\mathbb{R}^3}dx\,|\psi(x)|^4.
\end{equation}

\begin{teo}[BEC in the ground state]\label{thm: E0}\mbox{}\\
	Let $H_N$ be the Hamiltonian defined in \eqref{eq: H_N} and let $v_N$ be an interaction potential such that $ v $ satisfies Assumption \ref{asump: 1}. Assume that
	\begin{equation}
		\label{eq: gn conditions}
		\beta < \frac{1 }{ 3},	\qquad		\mbox{and}	\qquad		g_N \ll \min \left\{ N^{\frac{(1 - 3 \beta)(s+3)}{s+5}},  N^\frac{(s+3)\beta}{2(s+1)} \right\}.
	\end{equation}
	Let $\Psi $ be any normalized many-body state such that $ \langle \Psi, H_N \Psi \rangle \leq E_N + o(N) $  and let $\gamma_{\psi}$ be the associated reduced density matrix defined as in \eqref{eq: red dens}. Then,
	\begin{equation}\label{eq: BEC}
		1-\langle \varphi^{\mathrm{GP}}_N, \gamma_\Psi \varphi^{\mathrm{GP}}_N\rangle \leq Cg_N^{\frac{s+5}{s+3}}N^{3\beta-1} + Cg_N^{\frac{2(s+1)}{s+3}}N^{-\beta},
	\end{equation}
	where $\varphi^{\mathrm{GP}}_N\in L^2(\mathbb{R}^3)$ is the unique (up to a phase) normalized minimizer of \eqref{eq: EGP gN}.
\end{teo} 

\begin{remark}[Homogeneous trapping potential]
	\label{rem: asympt homo}
	For the sake of simplicity, we assumed that the trapping potential $ V_{\mathrm{ext}} $ in \eqref{eq: ext potential} is homogeneous of order $ s > 0$. As it is customary in the TF regime, such an assumption can be relaxed and replaced with the weaker condition that $ V_{\mathrm{ext}} $ is asymptotically homogeneous, i.e., $ V_{\mathrm{ext}}(x) \sim |x|^s $, as $ |x| \to + \infty $, because only the behavior at large $ |x| $ matters to leading order. This however makes the estimates of the error terms as well as the conditions on the parameters more implicit.
\end{remark}

\begin{remark} [GP minimizer]
\label{rem: GP minimizer}
We point out that, since the GP functional in \eqref{eq: EGP gN} depends on $ N $, the same holds for its minimizer $\varphi^{\mathrm{GP}}_N$. Indeed, it is very well known that the density $|\varphi^{\mathrm{GP}}_N|^2$ is pointwise close as $ N \to \infty $ to the minimizer of the TF functional
\begin{equation}
		\mathcal{E}^{\mathrm{trap}}_{\mathrm{TF},g_N}(\rho) = \int_{\mathbb{R}^3} d x \: \left\{ V_{\mathrm{ext}}(x) \rho + \frac{g_N }{ 2}(\smallint v) \rho^2 \right\}
\end{equation}
among positive densities $ \rho $ normalized in $ L^1 $. More precisely, one has (see, e.g., \cite[Thm. 2.2]{LSY})
\begin{equation}
	\label{eq: GP energy}
	\inf_{\left\| \psi \right\|_2 = 1} \mathcal{E}^{\mathrm{trap}}_{\mathrm{GP},N}(\psi) = \inf_{\rho \geq 0, \left\| \rho \right\|_1 = 1} \mathcal{E}^{\mathrm{trap}}_{\mathrm{TF},g_N}(\rho) \left(1 + o(1) \right) = g_N^{\frac{s}{s+3}} \inf_{\rho \geq 0, \left\| \rho \right\|_1 = 1} \mathcal{E}^{\mathrm{trap}}_{\mathrm{TF},1}(\rho) + O\left(g_N^{- \frac{2}{s+3}} \log g_N\right),
\end{equation}
as $ g_N \to + \infty $, and, denoting by $ \rho^{\mathrm{TF}}_{g_N} $ the unique TF minimizer, 
\begin{equation}
	g_N^{\frac{3}{s+3}} \rho^{\mathrm{TF}}_{g_N} \left( g_N^{\frac{1}{s+3}} x \right) = \rho^{\mathrm{TF}}_{1}(x),		\qquad		\left\| g_N^{\frac{3}{s+3}}\left|\varphi^{\mathrm{GP}}_N \left(g_N^{\frac{1}{s+3}} \cdot\right)\right|^2 - \rho^{\mathrm{TF}}_{1} \right\|_{\infty} = o(1),
\end{equation}
i.e., both the GP and TF minimizers spread over a length scale $ g_N^{\frac{1}{s+3}} $ (TF radius) and tend pointwise to $ 0 $ as $ N  \to + \infty $.
\end{remark}

\begin{remark} [Excitation spectrum] 
	In the present paper we do not investigate the excitation spectrum in the TF regime, but we remark that for a system of $N$ bosons confined in a box, it is possible to study it with similar techniques as in \cite{S}. The extension to the case of trapped bosons is however not completely trivial and beyond the scope of the present paper.
\end{remark}

\begin{remark}[Rotating systems] It would be interesting to generalize our result to rotating systems.
It is indeed well-know that if the rotational velocity $\Omega_{\mathrm{rot}}$ of the condensate is not too large, i.e., it is below a critical threshold $\Omega_c(N)$, the GP minimizer is unique. In this case, the proof of BEC is expected to hold with minor modifications. For faster rotation however the GP functional is known to have more than one minimizer. In that case the proof of BEC is more involved and one expects BEC in the mixed state given by convex combination of the GP minimizers (see \cite{LS,NRS}).
	
\end{remark}

As already anticipated above, after having proven that there is BEC in the ground state, a natural question that arises is whether BEC is preserved or not by time evolution. More explicitly, the question we would like to answer is whether we can find an \textit{effective dynamics} for a single particle state that well describes the evolution of the condensed state. We consider a system of $N$ trapped identical bosons, as described before, and we suppose that at the initial time the state is condensed. After a while, we remove the trap, thus the Hamiltonian of the system becomes
\[
	H_N = \sum_{j=1}^N (-\Delta_j) + g_N N^{3\beta-1}\sum_{1\leq j < k \leq N} v(N^\beta(x_j-x_k)),
\]
with domain $D(H_N):= L^2_s(\mathbb{R}^{3N})\cap H^2(\mathbb{R}^{3N})$.
The dynamics of the system is then encoded in the fact that the state at time $t$, $\psi_t \in L^2_s(\mathbb{R}^{3N})$, is a solution of the Schr\"odinger equation
\begin{equation}
	\label{eq:MB_schroedinger}
			\begin{cases}
				i\partial_t \psi_t	=H_N \psi_t,\\
				\psi_t\vert_{t=0}=\psi_0.
			\end{cases}
	\end{equation}

We prove that, up to some conditions on $\beta$ and $g_N$, we can find an  effective equation, i.e., the time-dependent GP equation  in the TF limit, which reads 
\begin{equation}
	i\partial_t \varphi_t =\big( -\Delta  + g_N\left|\varphi_t\right|^2\big)\varphi_t 
\end{equation}	
such that $\psi_t$ converges to the solution of the effective problem. Although we need some technical assumptions on the solution $\psi_t$, we are able to provide an explicit rate of convergence.

In the case of dilute limits, this question already has some answers. For example, it is well known \cite{ES07,ESY071,FK09,P1,RS09} that in the mean-field limit, if the initial state $\psi_0$ exhibits BEC on $\varphi_0$, then $\psi_t$  exhibits BEC on a family of single particle states $ \left\{\varphi_t\right\}_{t\in\mathbb R} $ which solve the {\itshape nonlinear Hartree equation}:
\begin{equation}\label{eq:classical_Hartree}
		i\partial_t\varphi_t = \left( -\Delta +v *	\left|\varphi_t \right|^2	\right)\varphi_t.
\end{equation}

A similar result for the GP regime is well know (see \cite{BOS,BS19,ESY071,ESY10,P2} and references therein). Meaning that if the initial datum exhibits BEC, then BEC is preserved also at later times on a family of single particle states $ \left\{\varphi_t\right\}_{t\in\mathbb R} $ which solve the {\itshape nonlinear GP equation}:
\begin{equation}\label{eq:classical_GP}
		i\partial_t\varphi_t = \big( -\Delta + 8\pi g\left|	\varphi_t \right|^2
\big)\varphi_t,
\end{equation}
where $g$ is a constant which is defined, in the limit as $N\rightarrow \infty$, as $N$ times the scattering length of the system, i.e.,  $Na\rightarrow g$ as $N\rightarrow \infty$. 

As we already discussed for the ground state energy, also for the dynamical framework there are different intermediate regimes corresponding to $ \beta\in\left(0,1\right) $. Also in this case \cite{P2} if there is BEC at the initial state, there is still BEC in the many-body state of the system at later times on a family of single particle states $ \left\{\varphi_t\right\}_{t\in\mathbb R} $ which solve the following nonlinear equation:
\begin{equation}\label{eq:intermediate_regime}
		i\partial_t\varphi_t = \big( -\Delta + (\smallint v) \left| \varphi_t \right|^2 \big)\varphi_t.
\end{equation}
Moreover, for the same regime, a norm approximation for the evolution of an initial state exhibiting BEC is also available (see \cite{BNNS}).
This last nonlinear equation is rather similar to the GP equation, with the major difference being that instead of having the constant $g$ here we find $\int v$. This is what one should expect since in this regime, $N$ times the  scattering length of the pair potential  goes to $\int v $ in the limit $N\rightarrow \infty$.

Thus, if we now look at our setting in which $Na = \mathcal{O}(g_N)$, we can expect that the effective dynamics is a cubic nonlinear Schr\"odinger equation. Since the effective scattering length of our potential grows as $ g_N $ as $N\rightarrow \infty$, such an equation can not be independent of $N$.

Moreover, we have that, in the sense of distributions, $ N^{3\beta} v\left(N^\beta x\right)\to \left(\int v\right)\delta_0\left(x\right) $ as $ N\to\infty $. Using this, one could also see equation \eqref{eq:intermediate_regime} as a limit for $ N\to\infty $ of an Hartree equation as in \eqref{eq:classical_Hartree} with the substitution $ v\to N^{3\beta} v\left(N^\beta\cdot\right) $. Notice also that this sort-of intermediate equation is still $ N $ dependent, but has now the advantage of being a good candidate for the equation governing the dynamics of the single particle state for our system.

Our discussion above led us to guess that we can expect condensation at later times only on a state which is still $ N-$dependent. For this purpose we will need to consider two different differential equations. In analogy to the GP limit, from now on we define the  GP equation as
\begin{equation}\label{eq:GP}
		i\partial_t\pgpt{t}	= \left( -\Delta + g_N\left(\smallint v\right) \left|	\pgpt{t} \right|^2\right)\pgpt{t}.
\end{equation}

In what follows we will need to work with the associated GP functional. Since to study the dynamics we switch off the trap, to distinguish this functional from the one introduced before, we define
\begin{equation}\label{eq: GP funct free}
	\mathcal{E}^{\mathrm{free}}_{\mathrm{GP}}(u) := \int_{\mathbb{R}^3}dx\,|\nabla u(x)|^2 + \frac{g_N }{ 2}(\smallint v)\int_{\mathbb{R}^3}dx\, |u(x)|^4.
	\end{equation}

\begin{teo}[Dynamics]\label{thm:GP_cond_main} Let $ \beta \in\left(0,1/6\right) $ and $ \lambda \in\left(3\beta, 1-3\beta\right) $.
	Let $ \psi_t $ be the solution of \eqref{eq:MB_schroedinger} and $ \varphi^{\mathrm{GP}}_{t} $ be the solution of \eqref{eq:GP} with initial data respectively $ \psi_0 $ and $ \varphi_0 $, both with norm equal to $ 1 $.
	Suppose that $ v $ satisfies Assumption \ref{asump: 1}. It holds true that 
	\begin{eqnarray}
		\left\|\gamma_{\psi_t}^{(1)}-|\pgpt{t}\rangle \langle \pgpt{t}|\right\|&\leq&\sqrt 2\left(
			N^\frac{1-\lambda}2
			\left\|
				\gamma_{\psi_0}^{(1)}
				-|\varphi_0\rangle\langle \varphi_0|
			\right\|^\frac12
			+N^{\frac{3\beta-\lambda}{2}}
		\right)
		e^{C_{v}C_N(\varphi_0, t)g_N\left|t\right|} \nonumber
	\\
		&&
		+C\sqrt{g_NN^{-\beta}}\left(1+ \mathcal{E}^{\mathrm{free}}_{\mathrm{GP}}(\varphi_0) + g_NN^{-\beta}\|\varphi_0\|_{L^\infty}^2\right)e^{C_v C_N(\varphi_0,t)^2g_N|t|}, \label{eq: dynamics}
	\end{eqnarray}
	where $\|\cdot\|$ denotes the operator norm and the constant $C_v$ only depends on the potential $v$ and $C_N(\varphi_0, t)$ is given by 
	\begin{equation}
	\label{eq: C varphi}
		C_N(\varphi_0, t)
		:=\left\|
			\varphi_0
		\right\|_{H^2\left(\mathbb R^3\right)}
		e^{Cg_N^2\left([\mathcal{E}^{\mathrm{free}}_{\mathrm{GP}}(\varphi_0)]^2 + g_N^2N^{-2\beta}\|\varphi_0\|_{L^\infty(\mathbb{R}^3)}^4\right)\left|t\right|}.
	\end{equation}
\end{teo}

\begin{remark}[Time-scale of the GP approximation]\label{rem: time}
	The initial datum determines the maximal time-scale on which the GP approximation makes sense through the quantity $ C_N(\varphi_0, t) $ and in particular the exponential appearing in \eqref{eq: C varphi}. This in turn depends on the GP energy of the initial datum via \eqref{eq: C varphi}, which, according to \eqref{eq: GP energy}, can not be smaller than $ C g_N^{s/(s+3)} $. We stress that, although we switch off the trap in studying the dynamics, heuristically $\varphi_0$ is the initial state of the trapped system: this is why here we suppose that $\mathcal{E}^{\mathrm{free}}_{\mathrm{GP}}(\varphi_0) \sim C g_N^{s/(s+3)}$. Now, the exponential on the right hand side of \eqref{eq: C varphi} is bounded by a constant
	as long as
	\begin{equation}
		\label{eq: tstar}
		t < t_{\star} \sim  g_N^{-\frac{2(2s+3)}{s+3}}.
	\end{equation}
	Note indeed that, assuming the $ H^2 $ norm of $ \varphi_0 $ to be bounded uniformly, the second summand in the exponent in \eqref{eq: C varphi} is much smaller than the first, i.e.,  of order $ g_N^2 N^{-2\beta} $, which vanishes because of the factor $ N^{-2\beta} $ (see also below).
	For a time shorter than $ t_{\star} $ now, $ C_N(\varphi_0, t) $ can be bounded by a constant.
	Plugging this bound into \eqref{eq: dynamics} and using similar arguments as above, one gets
	\begin{equation}
		\left\|
			\gamma_{\psi_t}^{(1)}
			-|\pgpt{t}\rangle \langle \pgpt{t}|
		\right\|
		\le C \left(
			N^\frac{1-\lambda}2
			\left\|
				\gamma_{\psi_0}^{(1)}
				-|\varphi_0\rangle\langle \varphi_0|
			\right\|^\frac12
			+ N^{\frac{3\beta-\lambda}{2}}
			+g_N^\frac{2s+3}{2(s+3)}N^{-\frac\beta2}
		\right),
	\end{equation}
	which tends to 0 as $ N \to +\infty $ provided the initial datum is not too far from a condensate on $ \varphi_0 $ and assuming
	\begin{equation}
		g_N \ll N^{\frac{\beta(s+3)}{2s+3}}.
	\end{equation}
	Note that, due to the kinetic energy being sub-leading with respect to the potential energy in the Thomas Fermi regime, to determine the optimal time scale for the validity of the GP approximation is not an easy task. 
\end{remark}

\begin{remark}[Hartree vs GP approximation]
	The estimate \eqref{eq: dynamics} is obtained via the following inequality 
	\[
		\left\|\gamma^{(1)}_{\psi_t} - |\varphi^{\mathrm{GP}}_t\rangle \langle \varphi^{\mathrm{GP}}_t|\right\| \leq \left\|\gamma^{(1)}_{\psi_t} - |\varphi^{\mathrm{H}}_t\rangle \langle \varphi^{\mathrm{H}}_t|\right\| + 2\left\|\varphi^{\mathrm{GP}}_t - \varphi^{\mathrm{H}}_t\right\|_2,
	\]
	where we denote by $\varphi^{\mathrm{H}}_t$ the solution of the Hartree equation (see \eqref{eq:Hartree}). More precisely, the first line in the right hand side of \eqref{eq: dynamics} is the estimate of $\left\|\gamma^{(1)}_{\psi_t} - |\varphi^{\mathrm{H}}_t\rangle \langle \varphi^{\mathrm{H}}_t|\right\|$ and the second line is a bound for $2\left\|\varphi^{\mathrm{GP}}_t - \varphi^{\mathrm{H}}_t\right\|_2$.
	Note then that our result suggests that $\beta = 1/7$ is the threshold from Hartree to GP. More precisely, if we assume condensation on the initial state $\varphi_0$ and $g_N N^{-\beta}\ll 1$ (see Remark \ref{rem: time} and Remark \ref{rem: vortex}), then for $\beta < 1/7$ we see that for values of $\lambda$ arbitrarily close to $1-3\beta$, Hartree gives a better estimates. However, if $\beta >1/7$, the term $\left\|\gamma^{(1)}_{\psi_t} - |\varphi^{\mathrm{H}}_t\rangle \langle \varphi^{\mathrm{H}}_t|\right\|$ gives a worse approximation than $\left\|\varphi^{\mathrm{GP}}_t - \varphi^{\mathrm{H}}_t\right\|_2$. Note also that the interval of times we can cover in Theorem \ref{thm:GP_cond_main} is the same for both the Hartree and the GP approximations since the constant $C_N(\varphi_0, t)$ in \eqref{eq: C varphi} is appearing in both of them.
\end{remark}

\begin{remark}[Time-scale of the vortex dynamics]
	\label{rem: vortex}
	As anticipated, it would be interesting to compare our result with the findings of \cite{JS}, concerning the effective dynamics of vortices\footnote{Note that, differently to the one of our paper, the setting in \cite{JS} is two-dimensional. This is indeed a natural framework to study the motion of vortices as two-dimensional objects. However, this setting can be thought as a system with a cylindrical trap. This justifies our comparison.}.

	However, a first major difference with the setting considered there is the absence of the trapping potential, which is mostly due to a technical obstruction for the propagation in time of suitable bounds along the nonlinear dynamics. This difference makes a naive comparison (i.e., for $ V_{\mathrm{ext}} = 0 $) not so meaningful, but in absence of an external trap, $ t_{\star} $ in \eqref{eq: tstar} is always much smaller than the characteristic time-scale of the vortex dynamics in \cite{JS}, i.e., $ t_{\mathrm{vortex}} \sim  g_N^{2/(s+3)}\log g_N $.
	However, if $ g_N $ is chosen to grow very slowly with respect to $ N $, i.e. $ g_N^3 \ll \log\log N $, then the time of the vortices can be achieved. 
\end{remark}

\noindent\textbf{Outline and sketch of the proofs} The paper is organized as follows. In Section \ref{sec: E0} we prove Theorem \ref{thm: E0}, whose goal is to prove BEC of the ground state $\Psi_0$ of $H_N$ on the one-particle state $\varphi^{\mathrm{GP}}_N$. The proof's strategy is very similar to the one about the mean-field scaling studied in \cite{GS,S}. More precisely,
\begin{itemize}
\item We first approximate the GP interaction with respect to the Hartree one, using some uniform bounds on $\varphi^{\mathrm{GP}}_N$ (see Proposition \ref{pro: L infty norms}).
\item Then, we follow the ideas already developed in \cite{GS,S}  to get a bound on the expected number of particles in the ground state which are not in $\varphi^{\mathrm{GP}}_N$ (see Proposition \ref{pro: N+}).
\item Finally, we use the aforementioned  bound to prove the condensation estimate. To do that we need to study, in the limit $N\rightarrow +\infty$, the bottom of the spectrum of the effective GP Hamiltonian as well as its spectral gap, this is done in Proposition \ref{pro: en gap}.
\end{itemize}
In Section \ref{sec: dynamics} we prove Theorem \ref{thm:GP_cond_main}. The main ideas are the following:
\begin{itemize}
	\item First we approximate the many body dynamics with the Hartree evolution. In this part we follow the ideas developed in \cite{P1} using also some Sobolev bounds on the Hartree state at time $t$ which are proved in Proposition \ref{pro: H2 norm}.
	\item Then, we compare the Hartree dynamics with the GP one using Gr\"onwall's lemma (see Proposition \ref{pro: from H to GP}).
\end{itemize}
In general, we consider a scaling limit in which the coupling constant, $g_N$, is $N$-dependent. This choice requires a careful analysis to keep track of the $N$-dependence in all the bounds we need to prove both Theorem \ref{thm: E0} and Theorem \ref{thm:GP_cond_main}.\\

\noindent\textbf{Acknowledgments.} We are grateful to Michele Correggi for his feedback and several insightful discussions. We also thank Peter Pickl for useful comments and discussions. The support of the National Group of Mathematical Physics (GNFM--INdAM) through Progetto Giovani 2016 ``Superfluidity and Superconductivity'' and Progetto Giovani 2018 ``Two-dimensional Phases'' is  acknowledged.
\section{Ground state energy}\label{sec: E0}
The aim of this section is to prove \cref{thm: E0}. We start by recalling that, under our assumptions, the GP functional in \eqref{eq: EGP gN} admits a unique positive minimizer, which for short we denote in this section by $\varphi^{\mathrm{GP}}$ (in place of $\varphi^{\mathrm{GP}}_N$) and which satisfies the Euler-Lagrange equation 
\begin{equation}\label{eq: GP eq}
	(-\Delta + V_{\mathrm{ext}} )\varphi^{\mathrm{GP}} + g_N (\smallint v)|\varphi^{\mathrm{GP}}|^2\varphi^{\mathrm{GP}} = \mu_{\mathrm{GP}}\varphi^{\mathrm{GP}},
\end{equation}
where
\begin{equation} \label{eq: def mu GP}
\mu_{\mathrm{GP}}:= E^\mathrm{trap}_{\mathrm{GP},N} +\frac{g_N}2(\smallint v) \|\varphi^{\mathrm{GP}}\|_{L^4\left(\mathbb R^3\right)} ^4.
\end{equation}
Because of the uniqueness of the solution to \eqref{eq: GP eq}, $\varphi^{\mathrm{GP}}$ is also the ground state of the $\mathrm{GP}$ operator $h^{\mathrm{GP}}$, defined as
\begin{equation}\label{eq: def hGP}
	h^{\mathrm{GP}} :=  -\Delta + V_{\mathrm{ext}}  + g_N(\smallint v)|\varphi^{\mathrm{GP}}|^2.
\end{equation}	
One can find a complete set of normalized eigenfunctions $\{\varphi_n\}_{n\in \mathbb{N}}$ for the GP operator $h^{\mathrm{GP}}$, where we identify $\varphi_0 \equiv \varphi^{\mathrm{GP}}$. Given that the spectrum of $h^{\mathrm{GP}}$ is discrete, we can assume that the  corresponding eigenvalues are ordered in such a way that $ \mu_{\mathrm{GP}}\equiv\mu_0<\mu_1 $ and $ \mu_j\le\mu_{j+1} $ for any $ j\ge 1 $. In what follows it is very important that the inequality $\mu_{\mathrm{GP}} < \mu_1$ is strict, which is due to the fact that $h^{\mathrm{GP}}$ has a unique ground state (see also Proposition \ref{pro: en gap}).

To prove \cref{thm: E0} it is convenient to introduce an operator, which we call $\mathcal{N}_+$, that counts the number of particles outside $\varphi^{\mathrm{GP}}$, i.e., the number of excited particles. For short, we now set $P := |\varphi^{\mathrm{GP}}\rangle \langle \varphi^{\mathrm{GP}}|$ and $Q := 1 - P = \sum_{n\neq 0}|\varphi_n\rangle\langle\varphi_n|$. Furthermore, $P_j$ and $Q_j$ will denote copies of the operators $P$ and $Q$ acting on the $j-$th particle. The operator $\mathcal{N}_+$ can then be explicitly written as  
\begin{equation}
 \mathcal{N}_+ := \sum_{j=1}^N Q_j.
\end{equation}

We now briefly  explain the strategy of the proof of \cref{thm: E0}. In order to prove BEC of the ground state $\Psi_0$ of $H_N$ in the one-particle state $\varphi^{\mathrm{GP}}$, we prove in \cref{pro: N+} an upper bound to $\langle \Psi_0, \mathcal{N}_+ \Psi_0\rangle$. To do that, we need to prove a lower bound on the energy gap $\mu_1 - \mu_{\mathrm{GP}}$.

\subsection{Preliminary Estimates} In this section we prove a lower bound for the energy gap $\mu_{\mathrm{GP}} - \mu_1$ and we estimate  $\|\varphi^{\mathrm{GP}}\|_{L^\infty\left(\mathbb R^3\right)}$ and $\|\nabla\varphi^{\mathrm{GP}}\|_{L^\infty\left(\mathbb R^3\right)}$ in terms of $ g_N $. 
\subsubsection{Energy gap}

Acting through the unitary transformation defined by  
\begin{equation}\label{eq: resc phi}
	\varphi(x) = \eps^{\frac{3 }{ s+2}}\psi\left(\eps^{\frac{2 }{ s+2}}x\right),  \qquad \eps := (g_N(\smallint v))^{-\frac{s+2 }{ 2(s+3)}}
\end{equation}
one gets
\begin{equation}\label{eq: resc Egp}
	\langle \varphi, h^{\mathrm{GP}}\varphi\rangle = \eps^{-\frac{2s }{ s+2}}\langle \psi, (-\eps^{2}\Delta + V_{\mathrm{ext}} + |\psi^{\eps}_0|^2)\psi\rangle =: \eps^{-\frac{2s }{ s+2}}\langle \psi, h_\eps\psi\rangle,
\end{equation}
where $\psi^\eps_0$ is the ground state of $h_\eps$ with energy $ \mu^\eps_0 $, which is explicitly given by
\begin{equation}\label{eq: def unitary transf}
	\varphi^{\mathrm{GP}}(x) =: \eps^{\frac{3 }{ s+2}}\psi^\eps_0\left ( \eps^{\frac{2 }{ s+2}}x \right).
\end{equation}
Analogously, we denote by $\mu^{\eps}_1>\mu^{\eps}_0 $ the second lowest eigenvalue\footnote{Note that the strict inequality $\mu^\eps_1 >\mu^\eps_0$ follows from the uniqueness  of the ground state of $h_\eps$.} of $h_\eps$, and by $\psi_1^\eps$ the associated normalized eigenstate. It easily follows from \eqref{eq: resc Egp} that 
\begin{equation}\label{eq: resc eps}
	\mu_1 - \mu_{\mathrm{GP}}  = \eps^{-\frac{2s }{ (s+2)}}(\mu^{\eps}_1 - \mu^{\eps}_{0} ).
\end{equation}
To study the energy gap, it is then sufficient to find a bound for $\mu^{\eps}_{1} -\mu^{\eps}_0$, that is what we do in the next proposition.

\begin{pro}[Energy gap]\label{pro: en gap}
	There exists a constant $C>0$ depending only on $ v $ and $ V_{\mathrm{ext}} $ such that 
	\begin{equation}\label{eq: en gap eff prob}
	 \mu^{\varepsilon}_1 - \mu^{\varepsilon}_0 > C\eps^2.
	\end{equation}
	As a consequence, there exists a constant $C>0$ depending only on $ v $ and $ V_{\mathrm{ext}} $ such that 
	\begin{equation}\label{eq: en gap}
		 \mu_1 - \mu_{\mathrm{GP}} > C g_N^{-\frac2{s+3}}.
	\end{equation}
\end{pro}
A key ingredient for the proof of Proposition \ref{pro: en gap} is an integral estimate for $\psi^{\eps}_1$. This is the content of the next proposition, whose proof is very similar to \cite[Proposition 3.5]{OR}.   
\begin{pro}[Integral decay estimates]\label{pro: int decay}
Let $A(x)$ be the \textit{Agmon distance} associated to $V_{\mathrm{ext}}(x) = k|x|^s$ $\mathrm{(}s\geq 2\mathrm{)} $, i.e., 
\begin{equation}\label{eq: agmon dist}
	A(x) = \sqrt{k}\frac{|x|^{1+\frac{s }{ 2}} }{ 1+s/2}.
\end{equation}
Then
	\begin{equation}\label{eq: integral decay}
			\int_{\mathbb{R}^3}\, dx\, e^{\frac{A(x) }{ \eps^2}}|\psi_1^\eps(x)|^2 \leq (1+\mu^\eps_1)e^{\frac{C }{ \eps^2}(1 + \mu_1^\eps) },
	\end{equation}
	where $C>0$ is a positive constant which depends only on the trapping potential $V_{\mathrm{ext}}$.
\end{pro}
In the proof we use the following lemma.
\begin{lem}\label{lem: negative} 
Let $W(x):= V_{\mathrm{ext}}(x) + |\psi_0^\eps(x)|^2 - \mu_1^\eps$ and let $\Phi(x) := A(x) / 2 \eps^2$, where $A(x)$ is as in \eqref{eq: agmon dist}. It holds true that 
\begin{equation} 
	\int_{\mathbb{R}^3}dx\, \left( W(x) - \eps^2|\nabla\Phi(x)|^2\right ) e^{2\Phi(x)}|\psi^\eps_1(x)|^2 \leq 0.
\end{equation}
\end{lem}
The proof of the lemma above is an adaptation of \cite[Lemma 3.6]{OR}. 
\begin{proof}[Proof of Proposition \ref{pro: int decay}]  We set 
	\begin{equation}\label{eq: def omega a}
	\Omega := \left\{x\in \mathbb{R}^3\, \vert\, V_{\mathrm{ext}}(x) - 4\mu^\eps_1 < 4 \right \}
\end{equation} 
	and
	\begin{equation}
		\Phi(x) := \frac{ 1}{ 2\eps^2}A(x),
	\end{equation}
so that 
\begin{equation}
	|\nabla\Phi(x)|^2 = \frac{1}{ 4\eps^2} V_{\mathrm{ext}}(x).
\end{equation}
We then have that for each $x\in \Omega^c \equiv \mathbb{R}^3\setminus \Omega$,
\begin{equation}\label{eq: bound W- nabla phi}
	(V_{\mathrm{ext}}(x) + |\psi_0^\eps(x)|^2 - \mu_1^{\eps}) - \eps^2|\nabla\Phi(x)|^2 \geq 1.
\end{equation}
In what follows we use $W(x)$ to denote the potential $W(x):=V_{\mathrm{ext}}(x) + |\psi_0^\eps(x)|^2 - \mu_1^{\eps}$. By \eqref{eq: bound W- nabla phi}, we get
\begin{multline}\label{eq: est 1 dec}
\int_{\Omega^c} dx\, e^{2\Phi(x)}|\psi^\eps_1(x)|^2 \leq \int_{\Omega^c}dx\, \left(W(x) -\eps^2|\nabla\Phi(x)|^2\right)e^{2\Phi(x)}|\psi^\eps_1(x)|^2 
 \\
 \leq- \int_{\Omega}dx\, \left(W(x) - \eps^2|\nabla\Phi(x)|^2\right) e^{2\Phi(x)}|\psi^\eps_1(x)|^2,
\end{multline}
where in the last inequality we used that 
\begin{equation}\label{eq: inte neg}
	\int_{\mathbb{R}^3}dx\, \left( W(x) - \eps^2|\nabla\Phi(x)|^2\right ) e^{2\Phi(x)}|\psi^\eps_1(x)|^2 \leq 0,
\end{equation}
 which is what is stated in Lemma \ref{lem: negative}.
Using now that $W(x) - \eps^2 |\nabla\Phi|^2 \geq -\mu^\eps_1$, from \eqref{eq: est 1 dec}, we have
\begin{equation}
	 \int_{\Omega^c}dx\, e^{2\Phi(x)}|\psi^\eps_1(x)|^2 \leq \mu^\eps_1\int_{\Omega}dx\, e^{2\Phi(x)}|\psi^\eps_1(x)|^2,
\end{equation}	
which implies
\begin{equation}
	 \int_{\mathbb{R}^3}dx\, e^{2\Phi(x)}|\psi^\eps_1(x)|^2 \leq (\mu^\eps_1 +1)\int_{\Omega}dx\, e^{2\Phi(x)}|\psi^\eps_1(x)|^2.
\end{equation}	
To conclude, we bound the integral in the r.h.s. above. It is then enough to use that there exists a constant $C>0$ which depends only on the trapping potential such that 
\begin{equation}
	\sup_{x\in \Omega} e^{\frac{\Phi(x) }{ \eps^2}} \leq e^{\frac{C }{ \eps^2}(1 + \mu_1^\eps) }.
\end{equation}
We then get that 
\begin{equation}
	 \int_{\mathbb{R}^3}dx\, e^{2\Phi(x)}|\psi^\eps_1(x)|^2 \leq (\mu^\eps_1 +1)e^{\frac{C }{ \eps^2}(1 + \mu_1^\eps) }.
\end{equation}
\end{proof}

\begin{pro}[Boundedness of $\mu^\eps_j$]\label{pro: bound mu} There exist a constants $C>0$ depending only on $ v $ and $ V_{\mathrm{ext}} $ such that, for $ j = 0,1 $,
	\begin{equation}
			\label{eq: unif bound mu}
			0 < \mu^\eps_j  \leq  C.
	\end{equation} 
\end{pro}
\begin{proof} We first prove that $\mu_0^\eps$ is bounded. It is trivial to see that the energy $ \mathcal{E}^\eps(\psi) $, where we set
	\begin{equation}
		\mathcal{E}^\eps(\psi):= \langle \psi, (-\eps^2\Delta + V_{\mathrm{ext}})\psi\rangle + \frac{1}{2}\int_{\mathbb{R}^3} dx\, |\psi(x)|^4,
	\end{equation} 
	is uniformly bounded, by evaluating it on a suitable trial state. Therefore, being $\mathcal{E}^\eps$ the energy functional associated to $h^{\eps}$ and being $\psi_0^\eps$ the corresponding ground state, one gets
	\begin{equation} 
		\mu^\eps_0 = \mathcal{E}^\eps(\psi^\eps_0) + \frac{1}{2}\|\psi^\eps_0\|_4^4 \leq 2\mathcal{E}^\eps(\psi^\eps_0) \leq C.
	\end{equation}
	
	We now bound $\mu^\eps_1$. By the min-max principle 
	\begin{equation}\label{eq: min max}
		\mu^\eps_1 = \inf_{\substack{\|\psi\|_{L^2\left(\mathbb R^3\right)} = 1 \\ \psi\perp \psi^\eps_0}}\langle \psi, h_\eps\psi\rangle.
 	\end{equation} 
 	Notice that given that $ V_{\mathrm{ext}} $ is even, then $ \psi^\eps_0 $ is even. Therefore, for any normalized odd trial state $ \psi_\mathrm{trial}\perp\psi^\eps_0 $, and as a consequence we get 
	\begin{equation}\label{eq: bound mu 1}
		\mu_1^\eps \leq \langle \psi_{\mathrm{trial}}, h_\eps\psi_{\mathrm{trial}}\rangle \leq C(1 + \|\psi^\eps_0\|_\infty^2).
	\end{equation}
	Furthermore, a direct application of the maximum principle to the eigenvalue equation
	\begin{equation}\label{eq: var eq phi 0 eps 2}
		\eps^2\Delta \psi^\eps_0= \left (V_{\mathrm{ext}} + |\psi^\eps_0|^2  - \mu_{0}^\eps\right)\psi_0^\eps
	\end{equation} 
	provides an upper bound on $ \|\psi^\eps_0\|_\infty $: since by elliptic theory $ \psi^\eps_0 $ is smooth, at any maximum point $ \bar{x} \in \mathbb{R}^3 $, one has
	\begin{equation} 
		0\geq \eps^2\Delta\psi^\eps_0(\bar{x}) \geq  \left(  |\psi^\eps_0(\bar{x})|^2  - \mu_0^\eps\right) \psi^\eps_0(\bar{x}),
	\end{equation}
	so that (recall that $ \psi^\eps_0>0 $)
	\begin{equation}
		|\psi^\eps_0(\bar{x})|^2 \leq \mu_0^\eps \leq C,
	\end{equation}
	which in turn implies the uniform boundedness of $\mu^\eps_1$ via \eqref{eq: bound mu 1}.
\end{proof}
\begin{proof}[Proof of Proposition \ref{pro: en gap}] We first prove \eqref{eq: en gap eff prob} and then deduce from it the estimate \eqref{eq: en gap}.\\

\underline{Bound for $\mu^{\varepsilon}_1 - \mu^{\varepsilon}_0$.} Let $u(x):= \psi^\eps_1(x)/\psi^\eps_0(x)$, which is well posed since $\psi^\eps_0(x)>0$. We have that
\begin{equation}
	\mu^\eps_1 = \mu^\eps_0 + \eps^2\int_{\mathbb{R}^3}dx\, |\nabla u(x)|^2 |\psi^\eps_0(x)|^2,\nonumber
\end{equation}
since, by the eigenvalue equations for $ \psi^\eps_0 $,
\begin{align}
	\int_{\mathbb{R}^3}dx\, \eps^2|\nabla\psi^\eps_1(x)|^2
	&= \int_{\mathbb{R}^3}dx\, \eps^2|u(x)\nabla\psi^\eps_0(x) + (\nabla u(x))\psi^\eps_0(x)|^2
\\
	&=\int_{\mathbb{R}^3}dx\, \left\{\eps^2|\nabla u(x)|^2|\psi^\eps_0(x)|^2- \eps^2|u(x)|^2 \psi^\eps_0(x) \Delta\psi^\eps_0(x)  \right \}\nonumber
\\
	&=\int_{\mathbb{R}^3}dx\
	\left\{
		\eps^2|\nabla u(x)|^2|\psi^\eps_0(x)|^2
		- V_\mathrm{ext}(x)|u(x)|^2 |\psi^\eps_0(x)|^2
	\right.\nonumber
\\
	&\qquad\left.
		-|u(x)|^2 |\psi^\eps_0(x)|^4
		+\mu_0^\eps|u(x)|^2 |\psi^\eps_0(x)|^2
	\right \}.\nonumber
\end{align}
To find a lower bound for $\mu^\eps_1 - \mu^\eps_0$, we then use Poincar\'e's inequality. More precisely, let $B_{R}:= \left\{ x\in \mathbb{R}^3\, \vert \, |x| < R \right\}$ for some $R>0$ to be fixed later and let
\begin{equation}
	u_{R} := \int_{B_R}u(x)|\psi^\eps_0(x)|^2.
\end{equation}
We can bound
\begin{equation}
	\mu^\eps_1 - \mu^\eps_0 
	\geq C_{R}\eps^2\int_{B_R}dx\, |u(x) - u_{R}|^2 |\psi^\eps_0(x)|^2,\nonumber
\end{equation}
	by the weighted Poincar\'e's inequality (see \cite{Ra} and references therein) and where $C_{R}>0$ is a finite constant depending only on $ R $. Now, recalling that $\psi^\eps_1\perp\psi^\eps_0$, we get
	\begin{eqnarray}
		\int_{B_R}dx\,|u(x) - u_{R}|^2 |\psi^\eps_0(x)|^2 &=& \int_{B_R}dx\, |\psi^\eps_1(x)|^2 +\left| \int_{B_R}dx\,\psi^\eps_0(x)\psi^\eps_1(x)\right|^2 \left(\int_{B_R} dx\,|\psi^\eps_0(x)|^2-2\right)
		\\
		&=& 1 - \int_{B_R^c}dx\, |\psi^\eps_1(x)|^2  - \left| \int_{B_R^c}dx\, \psi^\eps_0(x)\psi^\eps_1(x)\right|^2
\left(1+ \int_{B_R^c}dx\, |\psi^\eps_0(x)|^2 \right) \nonumber,
\end{eqnarray}
where we also used that $\|\psi^\eps_j\|_{L^2\left(\mathbb R^3\right)} =1$ for $j=0,1$ Using now Cauchy-Schwarz, we obtain
\begin{equation}
	\mu^\eps_1 - \mu^\eps_0 \geq C_{R} \eps^2  \left( 1 - 3\int_{B_R^c}dx\, |\psi^\eps_1(x)|^2 \right).
\end{equation}
We then have to bound 
\begin{equation} 
	\int_{B^c_R}\, dx\, |\psi^\eps_1(x)|^2,
\end{equation}
to do that we use Proposition \ref{pro: int decay}. Indeed, from \eqref{eq: integral decay}, we get that there exist two constants $C, C^\prime >0$ such that 
\begin{equation} 
	\int_{B^c_R}\, dx\, |\psi^\eps_1(x)|^2 \leq (1 + \mu_1^\eps)e^{-\frac{C}{\eps^2}(C^\prime R^{1 + \frac{s}{2}} - 1 - \mu^\eps_1)}.
\end{equation}

Given that $ \mu_1^\eps $ is uniformly bounded, there exists a constant $ C''>0 $ independent of $ \eps $ such that $ \mu_1^\eps\le C'' $. We can then fix $ R $ so that
\begin{equation} 
	R^{1 + \frac{s}{2}}
	>\frac{1+C''}{C'}
	\geq \frac{1+\mu^\eps_1}{C^\prime}, 
\end{equation}
to get
\begin{equation}
	\mu^{\eps}_1 - \mu^\eps_0 > C\eps^2.
\end{equation}

\underline{Bound for $\mu_{\mathrm{GP}} - \mu_1$.} By \eqref{eq: en gap eff prob} and \eqref{eq: resc eps}
\begin{equation} 
	\mu_{\mathrm{GP}} - \mu_1 \geq Cg_N^{-\frac{2}{s+3}},
\end{equation}
which implies the result.
\end{proof}

\subsubsection{Estimates for $\|\varphi^{\mathrm{GP}}\|_{L^\infty\left(\mathbb R^3\right)}$, $\|\nabla\varphi^{\mathrm{GP}}\|_{L^\infty\left(\mathbb R^3\right)}$}
\begin{pro}[$L^\infty$ bounds]\label{pro: L infty norms}
	Let $\varphi^{\mathrm{GP}}$ be the normalized ground state of the operator $h^{\mathrm{GP}}$ defined in \eqref{eq: def hGP}.There exists a positive constant $C>0$ such that
	\begin{equation}\label{eq: bounds infty}
		\|\varphi^{\mathrm{GP}}\|_{L^\infty\left(\mathbb R^3\right)} \leq Cg_N^{-\frac{3 }{ 2(s+3)}}, \qquad \|\nabla\varphi^{\mathrm{GP}}\|_{L^\infty\left(\mathbb R^3\right)}\leq Cg_N^{\frac{2s-3 }{ 2(s+3)}}.
	\end{equation}
\end{pro}
\begin{proof}
	Recall that $\varphi^{\mathrm{GP}}$ is such that $\langle \varphi^{\mathrm{GP}}, h^{\mathrm{GP}}\varphi^{\mathrm{GP}}\rangle = \mu_{\mathrm{GP}}$, it then follows that $\varphi^{\mathrm{GP}}$ satisfies the following variational equation 
	\begin{equation}\label{eq: var eq phigp}
		\Delta \varphi^{\mathrm{GP}} = \left (V_{\mathrm{ext}} + g_N(\smallint v)|\varphi^{\mathrm{GP}}|^2  - \mu_{\mathrm{GP}}\right)\varphi^{\mathrm{GP}}.
	\end{equation} 
	Proceeding similarly as in Proposition \ref{pro: bound mu}, one can prove that 
\begin{equation}
	|\varphi^{\mathrm{GP}}(x_0)|^2 \leq \frac{\mu_{\mathrm{GP}} }{ g_N(\smallint v)}.
\end{equation} 
From \eqref{eq: resc eps} together with \eqref{eq: unif bound mu}, it follows that 
\begin{equation}\label{eq: bound mugp}
	\mu^{\mathrm{GP}}\leq Cg_N^{\frac{s }{ s+3}},
\end{equation}
which immediately implies 
\begin{equation}\label{eq: bound phigp}
	\|\varphi^{\mathrm{GP}}\|_{L^\infty\left(\mathbb R^3\right)} \leq Cg_N^{-\frac{3 }{ 2(s+3)}}, \qquad s\geq 2 .
\end{equation}

 We now estimate the $L^\infty$-norm of $\nabla\varphi^{\mathrm{GP}}$. For short, we set 
 \begin{equation} 
	f(x):=V_{\mathrm{ext}}(x)\varphi^{\mathrm{GP}}(x) + g_N(\smallint v)|\varphi^{\mathrm{GP}}(x)|^2 \varphi^{\mathrm{GP}}(x) - \mu_{\mathrm{GP}}\varphi^{\mathrm{GP}}(x),
 \end{equation}
	so that $\varphi^{\mathrm{GP}}$ satisfies $\Delta\varphi^{\mathrm{GP}} = f$. 
The bound on the $L^\infty$-norm of $\nabla\varphi^{\mathrm{GP}}$ can then be proven by standard elliptic estimates. More precisely, there exists a constant $\gamma>0$ such that for all $x\in\mathbb{R}^3$ it holds
\begin{equation}\label{eq: elliptic estimate}
	\|\nabla\varphi^{\mathrm{GP}}\|_{L^\infty(B_{\rho/2}(x_0))}\leq  \gamma\| f\|_{L^\infty(B_\rho(x_0))} + \frac{16\gamma }{ \rho^4}\|\varphi^{\mathrm{GP}}\|_{L^1(B_\rho(x_0))},
\end{equation}
where we denote by $B_r(x_0):= \{y\in\mathbb{R}^3\,\vert\, |y-x_0| < r\}$. For a proof of \eqref{eq: elliptic estimate} we refer to (Proposition 11.2, \cite{DB}). By Cauchy-Schwarz inequality, we get 
\begin{equation} \label{eq: est grad inf 2}
	|\nabla\varphi^{\mathrm{GP}}(x_0)| \leq
	\|\nabla\varphi^{\mathrm{GP}}\|_{L^\infty(B_{\rho/2}(x_0))}\leq  C\left(\| f\|_{L^\infty(\mathbb{R}^3)} + \rho^{-\frac{5 }{ 2}}\|\varphi^{\mathrm{GP}}\|_{L^2(\mathbb{R}^3)}\right) \leq C\|f\|_{L^\infty(\mathbb{R}^3)},
\end{equation}
where in the last inequality we fix the radius $\rho$ to be $\rho= \|f\|_{L^\infty(\mathbb{R}^3)}^{-2/5}$ and used that $\|\varphi^{\mathrm{GP}}\|_{L^2(\mathbb{R}^3)} = 1$.
We then have to estimate the $L^\infty$ norm of $f$. To do that, we use the bounds \eqref{eq: bound mugp}, \eqref{eq: bound phigp} and we get 
 \begin{equation}\label{eq: bound f 1}
	\|f\|_{L^\infty(\mathbb{R}^3)} = \| V_{\mathrm{ext}}(x)\varphi^{\mathrm{GP}}(x) + g_N(\smallint v)|\varphi^{\mathrm{GP}}(x)|^2 \varphi^{\mathrm{GP}}(x) - \mu_{\mathrm{GP}}\varphi^{\mathrm{GP}}(x)\|_{L^\infty(\mathbb{R}^3)} \leq Cg_N^{\frac{2s-3 }{ 2(s+3)}} + \|V_{\mathrm{ext}}\varphi^{\mathrm{GP}}\|_{L^\infty(\mathbb{R}^3)}.
 \end{equation}
 Using the exponential decay of $\varphi^{\mathrm{GP}}$, which can be deduced from Proposition \ref{pro: int decay} together with the maximum principle for subharmonic functions and which implies that 
 \begin{equation} 
		|\varphi^{\mathrm{GP}}(x)| \leq C g_N^{-\frac{3}{2(s+3)}} e^{-Cg_N^{\frac{s+2}{s+3}}(|g_N^{-\frac{1}{s+3}}x| - C^\prime)}\qquad \mbox{for}\,\,\, |x| \geq g_N^{\frac{1}{s+3}} R, 
 \end{equation}
 for some $R>0$ big enough, it follows that  
 \begin{equation} \label{eq: bound f 2}
	\|V_{\mathrm{ext}}\varphi^{\mathrm{GP}}\|_{L^\infty(\mathbb{R}^3)} \leq Cg_N^{\frac{2s-3 }{ 2(s+3)}}.
 \end{equation}
 Inserting \eqref{eq: bound f 1} and \eqref{eq: bound f 2} in \eqref{eq: est grad inf 2}, we get
 \begin{equation} 
  \|\nabla\varphi^{\mathrm{GP}}\|_{L^\infty(\mathbb{R}^3)} \leq Cg_N^{\frac{2s-3 }{ 2(s+3)}}.
 \end{equation}
\end{proof}

\subsection{Bound for $\mathcal{N}_+$}Before proving Theorem \ref{thm: E0}, we need a bound on $\mathcal{N}_+$. This is the content of the next Proposition. The proof is very similar to the one of \cite[Lemma 1]{GS}.
In the proof we will use Lemma~\ref{lem: H vs GP}, which is proven below.

\begin{pro}[Bound for $\mathcal{N}_+$]\label{pro: N+}
	Let $ v $ satisfy \cref{asump: 1}. Then, as $N\rightarrow\infty$, there exists a constant $ C $ depending only on $ v $ such that
	\begin{equation}
		\langle \Psi_0, \mathcal{N}_+ \Psi_0 \rangle \leq C \left( g_N^{\frac{s+5}{s+3}} N^{3\beta} + g_N^{\frac{2(s+1)}{s+3}}N^{1-\beta}\right).
	\end{equation}
\end{pro}
\begin{proof}
We first prove an upper bound for $E_0(N)$. As a trial state we choose the one in which all the particles occupy the GP minimizer, i.e., ${(\varphi^{\mathrm{GP}})}^{\otimes N}$, obtaining
\begin{eqnarray}
	E_0(N) &\leq& N \langle \varphi^{\mathrm{GP}}, (-\Delta + V_{\mathrm{ext}})\varphi^{\mathrm{GP}}\rangle + g_N \frac{(N-1) }{ 2}\int_{\mathbb{R}^3} dx\, (v_N\ast|\varphi^{\mathrm{GP}}|^2)(x)|\varphi^{\mathrm{GP}}(x)|^2 \nonumber	\\
	&\leq& N \langle \varphi^{\mathrm{GP}}, (-\Delta + V_{\mathrm{ext}})\varphi^{\mathrm{GP}}\rangle + g_N(\smallint v)\frac{N-1 }{ 2}\int_{\mathbb{R}^3}dx\, |\varphi^{\mathrm{GP}}(x)|^4  \nonumber
	\\
	&&+ Cg_N\frac{N-1 }{ N^\beta}\|\varphi^{\mathrm{GP}}\|_{L^\infty\left(\mathbb R^3\right)} \|\nabla\varphi^{\mathrm{GP}}\|_{L^\infty\left(\mathbb R^3\right)}\nonumber 
	\\
	&\le& N E^{\mathrm{trap}}_{\mathrm{GP},N} +Cg_N N^{1-\beta}\|\varphi^{\mathrm{GP}}\|_{L^\infty\left(\mathbb R^3\right)} \|\nabla\varphi^{\mathrm{GP}}\|_{L^\infty\left(\mathbb R^3\right)}, \label{eq: upper bound on E_0}
\end{eqnarray}
where in the second inequality we used \cref{lem: H vs GP}.

We now prove the lower bound. For any $m\in\mathbb{N}$ we set $\psi_m(x) :=  |\varphi^{\mathrm{GP}}| ^ 2 (x) -\frac{1 }{ N}\sum_{j=1}^N \eta_m(x-x_j)$, where $\eta_m(x)\in C^\infty_c(\mathbb{R}^3)$ is a sequence of mollifiers, such that $\|\eta_m\|_2 = 1$ for any $m\in\mathbb{N}$ and $(v_N\ast \eta_m)(x)\rightarrow v_N(x)$ a.e. as $m\rightarrow \infty$. We then have
\begin{eqnarray}
	0\leq \int_{\mathbb{R}^3}dx\int_{\mathbb{R}^3}dy\, \overline{\psi}_m(x)v_N(x-y)\psi_m(y)& = & \int_{\mathbb{R}^3}dx\int_{\mathbb{R}^3}dy\,|\varphi^{\mathrm{GP}}(x)|^2 v_N(x-y)|\varphi^{\mathrm{GP}}(y)|^2 
	\\
	&& - \frac{2 }{ N}\sum_{j=1}^N \int_{\mathbb{R}^3}dx\int_{\mathbb{R}^3}dy\,|\varphi^{\mathrm{GP}}(x)|^2 v_N(x-y)\eta_m (y-x_j)\nonumber
	\\
	&&+ \frac{1 }{ N^2}\sum_{j,k=1}^N \int_{\mathbb{R}^3}dx\int_{\mathbb{R}^3}dy\,\eta_m(x-x_j)v_N(x-y)\eta_{m}(y- x_k),\nonumber
\end{eqnarray}	
where the first inequality follows from the fact that $\hat{v}\geq 0$.
Taking the limit $m\rightarrow \infty$, by the monotone convergence theorem, we get
\begin{equation} 
	0\leq  \int_{\mathbb{R}^3} dx\, (v_N\ast|\varphi^{\mathrm{GP}}|^2)(x)|\varphi^{\mathrm{GP}}(x)|^2 - \frac{2 }{ N}\sum_{j=1}^N (v_N\ast |\varphi^{\mathrm{GP}}|^2)(x_j) + \frac{1 }{ N^2}\sum_{j,k=1}^Nv_N(x_j - x_k),
\end{equation}
which implies that 
\begin{equation}\label{eq: bound potential}
	\frac{g_N }{ N}\sum_{1\leq j < k \leq N} v_N(x_j - x_k) \geq - \frac{g_N N }{ 2}\int_{\mathbb{R}^3} dx\, (v_N\ast|\varphi^{\mathrm{GP}}|^2)(x)|\varphi^{\mathrm{GP}}(x)|^2 + g_N\sum_{j=1}^N (v_N\ast |\varphi^{\mathrm{GP}}|^2)(x_j) - g_NN^{3\beta} v(0).
\end{equation}
From \eqref{eq: bound potential}, we then get
\begin{eqnarray}
	H^{\mathrm{trap}}_N  
	&\geq& \sum_{j=1}^N  \left[-\Delta_j + V_{\mathrm{ext}}(x_j)  + g_N (v_N\ast |\varphi^{\mathrm{GP}}|^2)(x_j)\right]\nonumber
	\\
	&& - \frac{g_N N }{ 2}\int_{\mathbb{R}^3} dx\, (v_N\ast|\varphi^{\mathrm{GP}}|^2)(x)|\varphi^{\mathrm{GP}}(x)|^2 - g_NN^{3\beta} v(0).\nonumber
\end{eqnarray}
Now, using Lemma \ref{lem: H vs GP} twice and exploiting the normalization of $ \varphi^{\mathrm{GP}} $, we can replace both convolutions, obtaining
\begin{equation}
	\label{eq: bound on h trap}
	H^{\mathrm{trap}}_N  - NE^{\mathrm{trap}}_{\mathrm{GP},N} \geq \sum_{j=1}^N (h_{j}^{\mathrm{GP}} - \mu_{\mathrm{GP}}) - g_N N^{3\beta}v(0) - Cg_N N^{1-\beta}\|\varphi^{\mathrm{GP}}\|_{L^\infty\left(\mathbb R^3\right)}\|\nabla \varphi^{\mathrm{GP}}\|_{L^\infty\left(\mathbb R^3\right)}.
\end{equation}
Hence, recalling the eigenfunctions $\{\varphi_n\}_{n\in\mathbb{N}}$  of $h^{\mathrm{GP}}$ with  eigenvalues $ \{\mu_n\}_{n\in\mathbb N} $ (with $ \varphi_0\equiv\varphi^\mathrm{GP} $, $ \mu_0\equiv\mu_\mathrm{GP} $), we have that
\begin{equation}\label{eq: bound hgp - mugp}
	\sum_{j=1}^N (h_{j}^{\mathrm{GP}} - \mu_{\mathrm{GP}}) = \sum_{j\geq 0} (\mu_j - \mu_0)|\varphi_j\rangle \langle\varphi_j| \geq (\mu_1 - \mu_0)\mathcal{N}_+ \geq 0
\end{equation}
that we plug into \eqref{eq: bound on h trap}, so obtaining in combination with the upper bound \eqref{eq: upper bound on E_0}
\begin{equation}\label{eq: est N+}
	(\mu_1 - \mu_{\mathrm{GP}})\langle \Psi_0, \mathcal{N}_+ \Psi_0\rangle \leq g_N N^{3\beta}v(0) + C g_N N^{1-\beta}\|\varphi^{\mathrm{GP}}\|_{L^\infty\left(\mathbb R^3\right)} \| \nabla \varphi^{\mathrm{GP}}\|_{L^\infty\left(\mathbb R^3\right)}.
\end{equation}
Applying \cref{pro: en gap} and \cref{pro: L infty norms}, we finally get
\begin{equation}
	Cg_N^{-\frac2{s+3}}\langle \Psi_0, \mathcal{N}_+ \Psi_0\rangle \leq (\mu_1 - \mu_{\mathrm{GP}})\langle \Psi_0, \mathcal{N}_+ \Psi_0 \rangle \leq C^\prime \left( g_N N^{3\beta} + g_N^{\frac{2s}{s+3}}N^{1-\beta}\right ),
\end{equation}
for two positive constant $ C, C' $. The result immediately follows.
\end{proof}

We now prove a Lemma on the comparison between the Hartree and GP interactions.

\begin{lem}[Hartree and GP interactions] \label{lem: H vs GP} Let $ v $ satisfy \cref{asump: 1}. Then, there exists a constant $C>0$ such that
	\begin{equation}
		\left \| v_N\ast |\varphi^{\mathrm{GP}}|^2(x) -\left (\smallint v \right ) |\varphi^{\mathrm{GP}}|^2(x) \right \|_{L^\infty\left(\mathbb R^3\right)} \leq \frac{C }{ N^\beta}\| \varphi^{\mathrm{GP}}\|_{L^\infty\left(\mathbb R^3\right)} \|\nabla \varphi^{\mathrm{GP}}\|_{L^\infty\left(\mathbb R^3\right)}
	\end{equation}
\end{lem}
\begin{proof}
A direct computation yields
\begin{align}
	&\left \| v_N\ast |\varphi^{\mathrm{GP}}|^2(x) -\left (\smallint v \right ) |\varphi^{\mathrm{GP}}|^2(x) \right \|_{L^\infty\left(\mathbb R^3\right)} = \sup_{x\in \mathbb{R}^3}\left | \int_{\mathbb{R}^3}dy\, v(y) \left (\left|\varphi^{\mathrm{GP}}(x-N^{-\beta}y)\right|^2 - |\varphi^{\mathrm{GP}}|^2(x)\right )\right |
	\\
	&\qquad\qquad\qquad\qquad\qquad\qquad= \sup_{x\in\mathbb{R}^3}\left|\int_{\mathbb{R}^3} dy\, v(y) \int_0^1 d\lambda\, \frac{2y }{ N^\beta}|\varphi^{\mathrm{GP}}(x-\lambda N^{-\beta}y)|  |\nabla \varphi^{\mathrm{GP}}(x-\lambda N^{-\beta}y)|\right|\nonumber
	\\
	&\qquad\qquad\qquad\qquad\qquad\qquad\leq \frac{2 }{ N^\beta}\left(\int_{\mathbb{R}^3} dy\, |y| |v(y)|\right ) \| \varphi^{\mathrm{GP}}\|_{L^\infty\left(\mathbb R^3\right)} \|\nabla \varphi^{\mathrm{GP}}\|_{L^\infty\left(\mathbb R^3\right)}.
\end{align}
\end{proof}

\subsection{Proof of \cref{thm: E0}}
By \cref{pro: N+} 
\begin{equation}
	\langle \Psi_0, \mathcal{N}_+ \Psi_0 \rangle \leq C \left( g_N^{\frac{s+5}{s+3}} N^{3\beta} + g_N^{\frac{2(s+1)}{s+3}}N^{1-\beta}\right),
\end{equation}
so that
\begin{align}
	1 - \langle \varphi^{\mathrm{GP}}, \gamma_{\Psi_0}\varphi^{\mathrm{GP}}\rangle &= \frac{1 }{ N}\left[ N - \langle \Psi_0, a^\ast(\varphi^{\mathrm{GP}})a(\varphi^{\mathrm{GP}})\Psi_0 \rangle \right] = \frac{1 }{ N}\langle \Psi_0, \mathcal{N}_+ \Psi_0\rangle\nonumber
\\
	&\leq Cg_N^{\frac{s+5}{s+3}} N^{3\beta - 1}+Cg_N^{\frac{2(s+1)}{s+3}}N^{-\beta},
\end{align}
which completes the proof.

\section{Dynamics}\label{sec: dynamics}
This section is devoted to the proof of Theorem \ref{thm:GP_cond_main}. We summarize here the main steps. 
\begin{itemize}
	\item First, we use the techniques introduced in \cite{P1} to approximate the many-body state at time $t$, $\psi_t $, in terms of $ \pht{t}$, i.e., in terms of the solution of the nonlinear Hartree equation which reads as
\begin{equation}\label{eq:Hartree}
		i\partial_t\pht{t} = \left( -\Delta + g_N v_N* \left| \pht{t} \right|^2\right)\pht{t},
\end{equation}
where we recall $v_N(x) = N^{3\beta} v(N^\beta x)$.
\item Afterwards, we estimate the distance between $ \pht{t} $ and the solution $ \pgpt{t} $ of the GP equation introduced in \eqref{eq:GP}, i.e., 
\[
	i\partial_t\varphi^{\mathrm{GP}}_t = (-\Delta + g_N(\smallint v)|\varphi^{\mathrm{GP}}_t|^2)\varphi^{\mathrm{GP}}_t	.
\]
\end{itemize}
In the following we will need to work with the Hartree energy functional, we set here the notation. We define
\begin{equation}\label{eq: H funct free}
\mathcal{E}^{\mathrm{free}}_{\mathrm{H}}(u) = \int_{\mathbb{R}^3}dx\,|\nabla u(x)|^2 + \frac{g_N }{ 2}\int_{\mathbb{R}^3}dx\, (v_N\ast|u|^2)(x)|u(x)|^2.
\end{equation}

 \subsection{Sobolev norms of the solutions}
 An important role in the proof of Theorem \ref{thm:GP_cond_main} is played the $L^\infty$ norms of $\varphi^{\mathrm{GP}}_t$ and $\varphi^{\mathrm{H}}_t$. We then adapt the proof of \cite[Proposition 3.1 (ii)]{BOS} to our framework to prove a bound for the $H^2$-norm of $\varphi^{\mathrm{GP}}$ and $\varphi^{\mathrm{H}}$, from which we deduce the estimates for the $L^\infty$-norms.
 \begin{pro}[Sobolev bound for $\varphi^{\mathrm{H}}_t$, $\varphi^{\mathrm{GP}}_t$]\label{pro: H2 norm}

	Let $ \pgpt{t} $ and $ \pht{t} $ be solutions of, respectively, \eqref{eq:GP} and \eqref{eq:Hartree}, both with initial datum $ \varphi_0\in H^n\left(\mathbb R^3\right) $. Assume that $ v $ satisfies Assumption \ref{asump: 1}.
	\begin{enumerate}[(i)]
	
	\item\label{itm: control kinetic energy}
		There exists a positive constant $C>0$ which depends only on the interaction $v$ such that
		\begin{equation}
		\label{eq: sup t GP sup t H}
			\sup_{t\in\mathbb R}
			\left\|\pgpt{t}\right\|_{H^1\left(\mathbb R^3\right)}\leq \mathcal{E}^{\mathrm{free}}_{\mathrm{GP}}(\varphi_0), \quad
			\sup_{t\in\mathbb{R}}\|\nabla\varphi^{\mathrm{H}}_t\|_2
			\leq C\mathcal{E}^{\mathrm{free}}_{\mathrm{GP}}(\varphi_0)
			+Cg_N N^{-\beta}\|\varphi_0\|_{L^\infty(\mathbb{R}^3)}^2.
		\end{equation}

	\item\label{itm: control higher norms pgpt}
		There exists a positive constant $ C $ such that for any $ t\in\mathbb R $
		\begin{align}
		\label{eq: sobolev growth for pht}
			&\left\|\pgpt{t}\right\|_{H^2\left(\mathbb R^3\right)}
			\le C\left\|\varphi_0\right\|_{H^2\left(\mathbb R^3\right)} e^{Cg_N^2[\mathcal{E}^{\mathrm{free}}_{\mathrm{GP}}(\varphi_0)]^2\left|t\right|},
		\\
		\label{eq: sobolev growth for pgpt}
			&\left\|\pht{t}\right\|_{H^2\left(\mathbb R^3\right)}
			\le C\left\|\varphi_0\right\|_{H^2\left(\mathbb R^3\right)} e^{Cg_N^2\left([\mathcal{E}^{\mathrm{free}}_{\mathrm{GP}}(\varphi_0)]^2 + g_N^2N^{-2\beta}\|\varphi_0\|_{L^\infty(\mathbb{R}^3)}^4\right)\left|t\right|}.
		\end{align}
		\end{enumerate}

\end{pro}
In the proof of Proposition \ref{pro: H2 norm}, we use some Strichartz estimates. In particular we use that for any function $ f\in L^\infty_tL^\frac65_x $, one has
\begin{align}
\label{eq: Strichartz}
	\sup_{t\in\left[0,T\right]}
	\left\|
		\int_0^t
		ds\
		e^{i\left(t-s\right)\Delta}
		f\left(s,\cdot\right)
	\right\|_{L^2\left(\mathbb R^3\right)}
	\le\sqrt T
	\sup_{t\in\left[0,T\right]}
	\left\|f\left(t,\cdot\right)\right\|_{L^\frac65\left(\mathbb R^3\right)}.
\end{align}
For a proof of \eqref{eq: Strichartz}, see \cite[Theorem 1.2]{KT}.
\begin{proof}[Proof of Proposition \ref{pro: H2 norm}]
	
	We start by showing the bounds in \eqref{eq: sup t GP sup t H}. The estimate for $\varphi^{\mathrm{GP}}_t$ directly follows from the fact that $ \widehat v\geq 0$ and that the energy is preserved by the GP dynamics.
We now look at $\varphi_t^{\mathrm{H}}$.
Notice that, proceeding similarly as in Lemma \ref{lem: H vs GP} and using Assumption \ref{asump: 1} we get
\begin{align}
	|\langle
		\varphi_0,
		(v_N \ast |\varphi_0|^2-|\varphi_0|^2)
		\varphi_0
	\rangle|
	\le
	Cg_N N^{-\beta}
	\left\|
		\nabla\varphi_0
	\right\|_{L^2\left(\mathbb R^3\right)}
	\left\|
		\varphi_0
	\right\|_{L^\infty\left(\mathbb R^3\right)}^2.
\end{align}

Using now Cauchy-Schwartz and the fact that the energy is preserved by the Hartree dynamics, we can write
\begin{equation}
	\|\nabla\varphi^{\mathrm{H}}_t\|_{L^2\left(\mathbb R^3\right)}^2\leq \langle \varphi_0, (-\Delta + \frac{g_N}2 v_N \ast |\varphi_0|^2)\varphi_0\rangle \leq \mathcal{E}^{\mathrm{free}}_{\mathrm{GP}}(\varphi_0)
	+C\left\|\nabla\varphi_0\right\|_{L^2\left(\mathbb R^3\right)}^2
	+Cg_N^2 N^{-2\beta}\|\varphi_0\|_{L^\infty(\mathbb{R}^3)}^4.
\end{equation}

We then get
\begin{equation}
	\sup_{t\in\mathbb{R}}\|\nabla\varphi^{\mathrm{H}}_t\|_2
	\leq C\sqrt{\mathcal{E}^{\mathrm{free}}_{\mathrm{GP}}(\varphi_0)}
	+Cg_N N^{-\beta}\|\varphi_0\|_{L^\infty(\mathbb{R}^3)}^2.
\end{equation}

We now prove \eqref{eq: sobolev growth for pgpt}. Being $ \pgpt{t} $ the solution  of \eqref{eq:GP}, we can write 
	\begin{align}
		\pgpt{t}
		=e^{it\Delta}\pgpt{t_0}
		-ig_N\int_{t_0}^t
		ds\
		e^{i\left(t-s\right)\Delta}
		\left(
			\left|\pgpt{t}\right|^2
			\pgpt{t}
		\right).
	\end{align}
	
	Consider now $ \alpha\in\mathbb N^3 $ with $ \left|\alpha\right|=2 $; we use the notation $ \partial^\alpha=\partial_{x_1}^{\alpha_1}\partial_{x_2}^{\alpha_2}\partial_{x_3}^{\alpha_3} $. We have
	\begin{align}
		\partial^\alpha\pgpt{t}
		&=e^{it\Delta}\partial^\alpha\pgpt{t_0}
		-ig_N\int_{t_0}^t
		ds\
		e^{i\left(t-s\right)\Delta}
		\partial^\alpha
		\left(
			\left|\pgpt{t}\right|^2
			\pgpt{t}
		\right)
	\\
		&=e^{it\Delta}\partial^\alpha\pgpt{t_0}
		-ig_N\int_{t_0}^t
		ds\
		e^{i\left(t-s\right)\Delta}
		\sum_{\gamma\le\beta\le\alpha}
		\binom\alpha\beta
		\binom\beta\gamma
		\left(
			\partial^{\left(\alpha-\beta\right)}
			\overline{\pgpt{t}}
			\partial^{\left(\beta-\gamma\right)}
			\pgpt{t}
			\partial^{\gamma}
			\pgpt{t}
		\right).\nonumber
	\end{align}
	By \eqref{eq: Strichartz}, we get
	\begin{align}
		&\sup_{t\in\left[t_0,t_0+T\right]}
		\left\|\partial^\alpha\pgpt{t}\right\|_{L^2\left(\mathbb R^3\right)}
		\le
		\left\|\partial^\alpha\pgpt{t_0}\right\|_{L^2\left(\mathbb R^3\right)}
	\\
		&\qquad\qquad+g_N
		\sup_{t\in\left[t_0,t_0+T\right]}
		\left\|
			\int_{t_0}^t
			ds\
			e^{i\left(t-s\right)\Delta}
			\sum_{\gamma\le\beta\le\alpha}
			\binom\alpha\beta
			\binom\beta\gamma
			\left(
				\partial^{\left(\alpha-\beta\right)}
				\overline{\pgpt{t}}
				\partial^{\left(\beta-\gamma\right)}
				\pgpt{t}
				\partial^{\gamma}
				\pgpt{t}
			\right)
		\right\|_{L^2\left(\mathbb R^3\right)}\nonumber
	\\
		&\qquad\le
		\left\|\partial^\alpha\pgpt{t_0}\right\|_{L^2\left(\mathbb R^3\right)}
		+64g_N\sqrt T
		\sum_{\gamma\le\beta\le\alpha}
		\sup_{t\in\left[t_0,t_0+T\right]}
		\left\|
			\partial^{\left(\alpha-\beta\right)}
			\pgpt{t}
			\partial^{\left(\beta-\gamma\right)}
			\pgpt{t}
			\partial^{\gamma}
			\pgpt{t}
		\right\|_{L^\frac65\left(\mathbb R^3\right)}.\nonumber
	\end{align}

	As in \cite[Proposition 3.1]{BOS}, via H\"older inequality and Sobolev embeddings, we get that there exists a constant $ C $ (which	depends only on the Sobolev embeddings) such that
	\begin{align}
		\sup_{t\in\left[t_0,t_0+T\right]}
		\left\|\partial^\alpha\pgpt{t}\right\|_{L^2\left(\mathbb R^3\right)}
		&\le
		\left\|\partial^\alpha\pgpt{t_0}\right\|_{L^2\left(\mathbb R^3\right)}
		+Cg_N\sqrt T
		\sup_{t\in\left[t_0,t_0+T\right]}
		\left\|
			\pgpt{t}
		\right\|_{H^1\left(\mathbb R^3\right)}^2
		\left\|
			\pgpt{t}
		\right\|_{H^2\left(\mathbb R^3\right)}.
	\end{align}
	
	This means, that up to increasing $ C $ and using \eqref{eq: sup t GP sup t H}, we get that for any $ t\in\left[t_0,t_0+T\right] $,
	\begin{align}
		\left\|\pgpt{t}\right\|_{H^2\left(\mathbb R^3\right)}
		\le\left\|\pgpt{t_0}\right\|_{H^2\left(\mathbb R^3\right)}
		+C E g_N\sqrt T
		\sup_{t\in\left[t_0,t_0+T\right]}
		\left\|
			\pgpt{t}
		\right\|_{H^2\left(\mathbb R^3\right)},\nonumber
	\end{align}
	where we denoted by $E \equiv \mathcal{E}^{\mathrm{free}}_{\mathrm{GP}}(\varphi_0)$.
	We can now fix $ T $ small enough so that $ C E g_N\sqrt T=1/2 $, as a consequence we get for any $ t\in\left[t_0,t_0+T\right] $
	\begin{align}
		\left\|\pgpt{t}\right\|_{H^2\left(\mathbb R^3\right)}
		&\le2\left\|\pgpt{t_0}\right\|_{H^2\left(\mathbb R^3\right)}.
	\end{align}
Repeating the argument, we get that for any $ t>0 $
	\begin{align}
		\left\|\pgpt{t+T}\right\|_{H^2\left(\mathbb R^3\right)}
		&\le2\left\|\pgpt{t}\right\|_{H^2\left(\mathbb R^3\right)}.
	\end{align}
	Fix now $ t>0 $ and let $ k\in\mathbb N $ such that $ t\in\left(\left(k-1\right)T,kT\right] $; then
	\begin{align}
		\left\|\pgpt{t}\right\|_{H^2\left(\mathbb R^3\right)}
		&\le2\left\|\pgpt{\left(k-1\right)T}\right\|_{H^2\left(\mathbb R^3\right)}
		\le2^k\left\|\varphi_0\right\|_{H^2\left(\mathbb R^3\right)}
		\le2^{\frac tT+1}\left\|\varphi_0\right\|_{H^2\left(\mathbb R^3\right)}
		=2^{4C^2E^2g_N^2t+1}\left\|\varphi_0\right\|_{H^2\left(\mathbb R^3\right)},
	\end{align}
	which implies \eqref{eq: sobolev growth for pht}.
		
	The proof for \eqref{eq: sobolev growth for pgpt} can be done in the same way. Proceeding as above, one gets
	\begin{align}
		\sup_{t\in\left[t_0,t_0+T\right]}
		\left\|\partial^\alpha\pht{t}\right\|_{L^2\left(\mathbb R^3\right)}
		\le&
		\left\|\partial^\alpha\pht{t_0}\right\|_{L^2\left(\mathbb R^3\right)}
	\\
		&
		+g_N\sqrt T
		\sum_{\gamma\le\beta\le\alpha}
		\binom\alpha\beta
		\binom\beta\gamma
		\sup_{t\in\left[t_0,t_0+T\right]}
		\left\|
			v_N*\left(
				\partial^{\left(\alpha-\beta\right)}
				\pht{t}
				\partial^{\left(\beta-\gamma\right)}
				\overline{\pht{t}}
			\right)
			\partial^{\gamma}
			\pht{t}
		\right\|_{L^\frac65\left(\mathbb R^3\right)}.\nonumber
	\end{align}
The only difference with respect to before is that to bound the interaction term one has to use Young's inequality for the convolution (see \cite[Theorem 4.2]{LL}).
\end{proof}

\subsection{From the many-body problem to Hartree.} 
In this section we want to compare the many-body problem with the Hartree approximation, the main result is in the proposition below.
\begin{pro}[From Many-Body to Hartree]\label{pro: from MB to H}
	Let $ \beta \in\left(0,\frac16\right) $ and $ \lambda \in\left(3\beta, 1-3\beta\right) $.
	Let $ \psi_t $ and $ \pht{t} $   be the normalized solution of \eqref{eq:MB_schroedinger} and \eqref{eq:Hartree} with initial data $ \psi_0 $ and $ \varphi_0 $, respectively.
	Suppose that $ v $ satisfies Assumption \ref{asump: 1}. We have
	\begin{equation}
	\label{eq:manybody_to_Hartree}
		\left\|
			|\pht{t}\rangle\langle\pht{t}|-\gamma_{\psi_t}^{\left(1\right)}
		\right\|
		\le\sqrt2\left(
			N^{\frac{1-\lambda }{ 2}}\left\||\varphi_0\rangle\langle\varphi_0|-\gamma_{\psi_0}^{\left(1\right)}\right\|^\frac12
			+N^\frac{3\beta}2
		\right)
		e^{C_{v}(C_N(\varphi_0,t) + C_N(\varphi_0,t)^2)g_N\left|t\right|},
	\end{equation}
	where	 $ C_{v}$ is a constant depending on the $ L^1 $ and the $ L^2 $ norms of $ v $, and $C_N(\varphi_0,t)$ is given by
	
	\begin{equation}
		C_N(\varphi_0,t)=
		\left\|
			\varphi_0
		\right\|_{H^2\left(\mathbb R^3\right)}
		e^{Cg_N^2\left([\mathcal{E}^{\mathrm{free}}_{\mathrm{GP}}(\varphi_0)]^2 + g_N^2N^{-2\beta}\|\varphi_0\|_{L^\infty(\mathbb{R}^3)}^4\right)\left|t\right|}.
	\end{equation}
\end{pro}

The key idea for the proof of Proposition \ref{pro: from MB to H} is to control the number of {\itshape bad particles} in the many-body system  (i.e. the particles not in the state $ \pht{t} $), we do that using methods introduced in \cite{P1}.

We now fix the notation for the projectors on the spaces of good and bad particles. We collect and prove the main general properties of these operators in \cref{app: projector}.

\begin{defi}\label{def:projectors}

	Let $ \varphi \in L^2\left(\mathbb R^3\right) $ and $ \psi\in L^2_\mathrm s\left(\mathbb R^{3N}\right) $.

	\begin{enumerate}

	\item
		For any $ 1\le j\le N $, the projectors $ p_j^\varphi: L^2\left(\mathbb R^{3N}\right)\to L^2\left(\mathbb R^{3N}\right) $ represent the probability of the $ j $-th particle of the state $ \psi $ of being in the state $ \varphi $, i.e. 
		\begin{align}\label{eq: def pj qj}
			p_j^\varphi\psi\left(x_1,\ldots,x_N\right)
			:=&\left(
				|\varphi\rangle\langle\varphi|
			\right)_j\psi\left(x_1,\ldots,x_N\right)
		\\
			\equiv&\varphi\left(x_j\right)
			\int_{\mathbb R^3}
			dz\
			\varphi^*\left(z\right)
			\psi\left(x_1,\ldots,x_{j-1},z,x_{j+1},\ldots,x_N\right).\nonumber
		\end{align}		
		Analogously, $ q_j^\varphi:=1-p_j^\varphi $ represents the probability of the $ j $-th particle of the state $ \psi $ not being in the state $ \varphi $.

	\item 
		For any $ 0\le k\le j\le N $ we set
		\begin{equation}
			\mathcal{A}_k^j:
			=\left\{a:=\left(a_1, a_2,\ldots, a_j\right)
			\in\left\{0,1\right\}^j
			:\
			\sum_{l=1}^j a_l =k\right\}
		\end{equation}
		and define the orthogonal projector $ P_{j, k}^\varphi $ on $ L^2\left(\mathbb R^{3N}\right) $ as
		\begin{equation}
			P_{j, k}^\varphi:
			=\sum_{a \in \mathcal{A}_k^j}
			\prod_{l=N-j+1}^N
			\left[
				\left(
					p_l^\varphi
				\right)^{1-a_l}
				\left(
					q_l^\varphi\
				\right)^{a_l}
			\right].
		\end{equation}
		Note that the rank of $ P_{j,k}^\varphi $ is the space of states that in the last $ j $ particles have exactly $ k $ that are not in the state $ \varphi $ (we sum over $ \mathcal A_k^j $ to make the state symmetric). In the following we denote by $P_{k}^{\varphi}$ the projector given by $P_{N,k}^\varphi$.
	\item\label{itm:transf_of_f}
		For any function $ f:\left\{0,\ldots,N\right\}\to \mathbb R $ we define the operator $ \widehat f^\varphi: L^2\left(\mathbb R^{3N}\right)\to L^2\left(\mathbb R^{3N}\right) $ as
		\begin{equation}
		\label{eq:hat}
			\widehat{f}^\varphi:
			=\sum_{k\in\mathbb Z} f\left(k\right)
			P_k^\varphi.
		\end{equation}
		where $ f \left(k\right)=0 $ for any $ k\in\mathbb Z\setminus\left\{0,\ldots,N\right\} $.

		Similarly, if $ f_d\left(k\right) :=f\left(k+d\right) $ for any $d, k\in\mathbb Z $, we define the operator $ \widehat f_d^\varphi:L^2\left(\mathbb R^{3N}\right)\to L^2\left(\mathbb R^{3N}\right) $ associated to $ f_d $ as
		\begin{equation}
			\widehat{f}^\varphi_d:
			=\sum_{k \in\mathbb{Z}}
			f\left(k+d\right) P_k^\varphi.
		\end{equation}
		
		Throughout the following Section, the hat $ \widehat{\cdot} $ will solely be used in the sense of Definition \ref{def:projectors}.
		
	\end{enumerate}

\end{defi}

	Following \cite{P1}, we now introduce a functional which encodes the fraction of condensation within itself. 

	\begin{defi}
	\label{def:weight_fnct}
	
		For any $ \lambda\in\left(0,1\right) $ we define the weight function $ \mu^\lambda:\left\{0,\ldots,N\right\}\to\mathbb R $ as
		\begin{equation}
			\mu^\lambda\left(k\right):
			=\left\{
				\begin{array}{ll}
					\frac k{N^\lambda}, & \mbox{for } k\le N^{\lambda} ,\\
					1, & \mbox{otherwise.}
				\end{array}
			\right.
		\end{equation}
		
		Moreover, for any $ N \in \mathbb{N} $, we define the functional $ \alpha_N^\lambda:L_\mathrm s^2\left(\mathbb R^{3N}\right)\times L^2\left(\mathbb R^{3N}\right)\to \mathbb R $ as\footnote{For ease of notation we will write $ \widehat\mu^{\lambda,\varphi} $ instead of $ \widehat{\mu^\lambda}^\varphi $.}
		\begin{equation}
			\alpha_N^\lambda\left(\psi,\varphi\right):
			=\left\langle
				\psi, 
				\widehat\mu^{\lambda,\varphi}
				\psi
			\right\rangle
			=\left\|
				\left(\widehat\mu^{\lambda,\varphi}\right)^{1/2}
				\psi
			\right\|_{L_\mathrm s^2\left(\mathbb R^{3N}\right)}^2.
		\end{equation}
	
	\end{defi}
	
The proof of \cref{pro: from MB to H} is based on some Grönwall estimate on the quantity $\alpha^\lambda_N$, which we prove in the next proposition.
	
	\begin{pro}[Grönwall estimate on $\alpha^\lambda_N$]
	\label{pro: Gronwall alpha}
	
		Let $ \beta \in\left(0,\frac16\right) $ and $ \lambda \in\left(3\beta, 1-3\beta\right) $.
		Let $ \psi_t $ and $ \pht{t} $   be the normalized solution of \eqref{eq:MB_schroedinger} and \eqref{eq:Hartree} with initial data $ \psi_0 $ and $ \varphi_0 $, respectively.
		Suppose that $ v $ satisfies Assumption \ref{asump: 1}; we then have
		\begin{equation}
		\label{eq:main_Grönwall}
			\alpha_N^\lambda\left(\psi_t,\pht{t}\right)
			\le\left(
				\alpha_N^\lambda\left(\psi_0,\varphi_0\right)
				+\frac1{N^{\lambda-3\beta}}
			\right)e^{C_{v}(C_N(\varphi_0,t) + C_N(\varphi_0,t)^2)g_N|t|},
		\end{equation}
		where	 $ C_{v} $ is a constant depending on the $ L^1 $ and the $ L^2 $ norms of $ v $, and $C_N(\varphi_0,t)$ is given by 
		\begin{equation}
			C_N(\varphi_0,t)=
			\left\|
				\varphi_0
			\right\|_{H^2\left(\mathbb R^3\right)}
			e^{Cg_N^2\left([\mathcal{E}^{\mathrm{free}}_{\mathrm{GP}}(\varphi_0)]^2 + g_N^2N^{-2\beta}\|\varphi_0\|_{L^\infty(\mathbb{R}^3)}^4\right)\left|t\right|}.
		\end{equation}
	\end{pro}

\subsubsection{Proof of \cref{pro: Gronwall alpha}} To prove \cref{pro: Gronwall alpha} we need some preliminary results which are stated in \cref{lem:derivative_alpha} and in \cref{pro:terms_of_alpha}. In the proof of these two intermediate results, we use some operator estimates proven in \cref{app: projector}. To simplify notation we will use the following definition.

\begin{defi}
We denote by $ U_{j,k} $ the difference between the pair interaction and the mean-field interactions for two particles $ j $ and $ k $, i.e.
\begin{equation}
	U_{j,k}:=
	\left(N-1\right)v_N\left(x_j-x_k\right)
	-N v_N*|\pht{t}|^2\left(x_j\right)
	-N v_N*|\pht{t}|^2\left(x_k\right).
\end{equation}

\end{defi}

\begin{lem}
\label{lem:derivative_alpha}

Let $ \lambda\in\left(0,1\right) $. Let $ \psi_t $ and $ \pht{t} $   be the normalized solution of \eqref{eq:MB_schroedinger} and \eqref{eq:Hartree}, respectively. Then we have,
\begin{equation}
	\partial_t\alpha_N^\lambda\left(\psi_t,\pht{t}\right)
	=\Gamma_N^\lambda\left(\psi_t,\pht{t}\right),
\end{equation}
where $ \Gamma_N^\lambda:L_\mathrm s^2\left(\mathbb R^{3N}\right)\times L^2\left(\mathbb R^3\right)\to\mathbb R $ is defined as
\begin{eqnarray}
	\Gamma_N^\lambda\left(\psi,\varphi\right)
	&:= &2g_N\Im\left(\left\langle
		\psi,
		\left(
			\widehat\mu^{\lambda,\varphi}
			-\widehat\mu_1^{\lambda,\varphi}
		\right)
		p_1^\varphi p_2^\varphi
		U_{1,2}
		p_1^\varphi q_2^\varphi
		\psi
	\right\rangle\right)
\\
	&&+g_N\Im\left(\left\langle
		\psi,
		\left(
			\widehat\mu^{\lambda,\varphi}
			-\widehat\mu_2^{\lambda,\varphi}
		\right)
		p_1^\varphi p_2^\varphi
		U_{1,2}
		q_1^\varphi q_2^\varphi
		\psi
	\right\rangle\right)\nonumber
\\
	&&+2g_N\Im\left(\left\langle
		\psi,
		\left(
			\widehat\mu^{\lambda,\varphi}
			-\widehat\mu_1^{\lambda,\varphi}
		\right)
		p_1^\varphi q_2^\varphi
		U_{1,2}
		q_1^\varphi q_2^\varphi
		\psi
	\right\rangle\right).\nonumber
\end{eqnarray}

\end{lem}

\begin{proof}

	First notice that for any $ 1\le j\le N $ if we differentiate $ p_j^\pht{t} $ we get
	\begin{equation}
		\partial_tp_j^{\pht{t}}=\left( |\pht{t}\rangle\langle\partial_t\pht{t}|\right)_j + \left(|\partial_t\pht{t}\rangle\langle\pht{t}|\right)_j=i\left[p_j^\pht{t},-\Delta_j+g_Nv_N*\left|\pht{t}\right|^2\left(x_j\right)
		\right],
	\end{equation}
	and similarly for $ q_j^\pht{t} $
	\begin{equation}
		\partial_tq_j^{\pht{t}}=i\left[q_j^\pht{t},-\Delta_j+g_Nv_N*\left|\pht{t}\right|^2\left(x_j\right)\right].
	\end{equation}
	It is then useful to introduce $ H_\mathrm H^\varphi $ defined as
	\begin{equation}
		H_\mathrm H^\varphi:
		=\sum_{j=1}^N\left[-\Delta_j+g_Nv_N*\left|\varphi\right|^2\left(x_j\right)
		\right].
	\end{equation}
	As a consequence, for any $ 1\le j\le N $, we can write
	\begin{equation}
		\partial_tp_j^{\pht{t}}
		=i\left[
			p_j^{\pht{t}},
			H_\mathrm H^{\varphi^{\mathrm{H}}_t}
		\right],\qquad \partial_tq_j^{\pht{t}}
		=i\left[
			q_j^{\pht{t}},
			H_\mathrm H^{\varphi^{\mathrm{H}}_t}
		\right].
	\end{equation}
	Given that for any $ 0\le k\le j\le N $ the operator $ P_{j,k}^\pht{t} $ is just a linear combination of $ p_l^\pht{t} $'s and $ q_l^\pht{t} $'s, we deduce that for any weight $ f:\left\{0,\ldots,N\right\}\to\mathbb R^+ $ we get
	\begin{equation}
	\label{eq:deriv_of_f}
		\partial_t\widehat{f}^\pht{t}
		=i\left[
			\partial_t\widehat{f}^\pht{t},
			H_\mathrm H^\pht{t}
		\right].
	\end{equation}
	In order to simplify the notation, for the rest of the proof we will drop the labels $ \pht{t} $ and $ \lambda $ in the formulas involving the weights and the effective many-body nonlinear operator. 
	If now we differentiate $ \partial_t\alpha_N^\lambda\left(\psi_t,\pht{t}\right) $, using \eqref{eq:deriv_of_f} and the symmetry of $ \psi_t $ and $ \widehat\mu\psi_t $, we get
	\begin{equation}\label{eq: deriv t alpha}
		\partial_t\alpha_N^\lambda\left(\psi_t,\pht{t}\right)
		=-i\left\langle
			\psi_t,
			\left[
				\widehat\mu,
				H_N-H_\mathrm H
			\right]
			\psi_t
		\right\rangle
	=-i\frac{g_N}2
		\left\langle
			\psi_t,
			\left[
				\widehat\mu,
				U_{1,2}
			\right]
			\psi_t
		\right\rangle =  g_N\Im\left(
			\left\langle
				\psi_t,
				\widehat\mu\
				U_{1,2}
				\psi_t
			\right\rangle
		\right).
	\end{equation}
	We want to now decompose the action of the operator $ \widehat\mu\ U_{1,2} $ in terms of $ p_1 $, $ p_2 $, $ q_1 $ and $ q_2 $.
	To do so, recall that $ P_{j,k} $ projects on the space of states in which in the last $ j $ particles there are exactly $ k $ particles that are not in the state $ \pht{t} $; we can then write
	\begin{equation}
		P_k\equiv P_{N,k}
		 = p_1p_2P_{N-2,k}
		+\left(p_1q_2+q_1p_2\right)P_{N-2,k-1}
		+q_1q_2P_{N-2,k-2}.\nonumber
	\end{equation}
	As a consequence, for any weight $ f:\left\{0,\ldots, N\right\}\to\mathbb R $ we can use the equality above to get
	\begin{eqnarray}\label{eq: hat f Pk}
		\widehat f
		&=&\sum_{k\in\mathbb Z}
		f\left(k\right)P_k
	=\sum_{k\in\mathbb Z}
		f\left(k\right)\left(
			p_1p_2P_{N-2,k}
			+\left(p_1q_2+q_1p_2\right)P_{N-2,k-1}
			+P_{N-2,k-2}
		\right)\nonumber
	\\
		&&-p_1p_2\sum_{k\in\mathbb Z}f\left(k+2\right)P_{N-2,k}
		-\left(p_1q_2+q_1p_2\right)\sum_{k\in\mathbb Z}f\left(k+1\right)P_{N-2,k-1}\nonumber
	\\
		&=&p_1p_2(
			\widehat f
			-\widehat f_2
		)
		+\left(
			p_1q_2
			+q_1p_2
		\right)(
			\widehat f
			-\widehat f_1
		)
		+\sum_{k\in\mathbb Z}f\left(k\right)P_{N-2,k-2}.
	\end{eqnarray}
	We can then insert \eqref{eq: hat f Pk} in \eqref{eq: deriv t alpha} and we get 
	\begin{align}
		\partial_t\alpha_N^\lambda\left(\psi_t,\pht{t}\right)
		&=g_N\Im\left(
			\left\langle
				\psi_t,
				\left(
					p_1p_2\left(
						\widehat\mu
						-\widehat\mu_2
					\right)
					+2p_1q_2\left(
						\widehat\mu
						-\widehat\mu_1
					\right)
				\right)
				U_{1, 2}
				\psi_t
			\right\rangle
		\right)
	\\
	&= \sum_{\substack{\sharp_1 = p_1, q_1 \\ \sharp_2 = p_2, q_2}} g_N\Im\left(\langle \psi_t, p_1p_2 (\hat{\mu} -\hat{\mu}_2)U_{1,2} \sharp_1\sharp_2\psi_t\rangle\right) + g_N\Im\left(\langle \psi_t, p_1q_2 (\hat{\mu} -\hat{\mu}_1)U_{1,2} \sharp_1\sharp_2\psi_t\rangle\right),\nonumber
 	\end{align}
	where we used that the last term in \eqref{eq: hat f Pk} is self-adjoint and commutes with $ U_{1,2} $.
	It turns out that many of those terms are either identical or vanishing. In particular, we can use \eqref{itm:commuting_hats} to have that $p_1p_2
			\left(
				\widehat\mu
				-\widehat\mu_2
			\right)
			U_{1, 2}
			p_1p_2
			=\left(
				p_1p_2
				\left(
					\widehat\mu
					-\widehat\mu_2
				\right)
				U_{1, 2}
				p_1p_2
			\right)^\dagger$,
	which implies that $\Im\left(\langle \psi_t, p_1p_2 (\hat{\mu} -\hat{\mu}_2)U_{1,2} p_1p_2\psi_t\rangle\right) =0$ and similarly we get that $\Im\left(\langle \psi_t, p_1q_2 (\hat{\mu} -\hat{\mu}_1)U_{1,2} p_1q_2\psi_t\rangle\right) =0$. Moreover, using the symmetry of $\psi_t$, an easy computation shows that $\langle \psi_t, p_1q_2 (\hat{\mu} - \hat{\mu}_1) U_{1,2}q_1p_2\psi_t\rangle \in \mathbb{R}$, therefore the imaginary part is vanishing. For the remaining terms, we can use again 
 \eqref{itm:commuting_hats} and after some manipulations we can combine them to get $g_N\mathrm{Im}\langle \psi_t, p_1p_2 (\hat{\mu} - \hat{\mu}_1)U_{1,2}p_1 q_2 \psi_t\rangle.$ %
	We then obtain
	\begin{eqnarray}
		\partial_t\alpha_N^\lambda\left(\psi_t,\pht{t}\right)
		&=&2g_N\Im\left(
			\left\langle
				\psi_t,
				p_1p_2
				\left(
					\widehat\mu
					-\widehat\mu_1
				\right)
				U_{1, 2}
				p_1q_2
				\psi_t
			\right\rangle
		\right)
	\\
		&&+g_N\Im\left(
			\left\langle
				\psi_t,
				p_1p_2
				\left(
					\widehat\mu
					-\widehat\mu_2
				\right)
				U_{1, 2}
				q_1q_2
				\psi_t
			\right\rangle
		\right)\nonumber
	\\
		&&+2g_N\Im\left(
			\left\langle
				\psi_t,
				p_1q_2
				\left(
					\widehat\mu
					-\widehat\mu_1
				\right)
				U_{1, 2}
				q_1q_2
				\psi_t
			\right\rangle
		\right),\nonumber
	\end{eqnarray}
	which concludes the proof.
\end{proof}
The next step is to estimate separately the three terms appearing in $ \Gamma_N^\lambda\left(\psi_t,\pht{t}\right) $ in terms of $ \alpha_N^\lambda\lambda\left(\psi_t,\pht{t}\right) $ in order to prove a Grönwall-type estimate for $ \alpha_N^\lambda\lambda\left(\psi_t,\pht{t}\right) $. This is the aim of the next Proposition.

\begin{pro}\label{pro:terms_of_alpha}

	Let $ \lambda\in\left(0,1\right) $, let $ \psi_t $ and $ \varphi_t $ be normalized solutions of \eqref{eq:MB_schroedinger} and  \eqref{eq:Hartree}, respectively. The following estimates hold true.
	
	\begin{align}
	\label{eq:oneq}
		\left|
			\Im\left(\left\langle
				\psi_t,
				\left(
					\widehat\mu_1^{\lambda,\pht{t}}
					-\widehat\mu^{\lambda,\pht{t}}
				\right)
				p_1^\pht{t} p_2^\pht{t}
				U_{1,2}
				p_1^\pht{t} q_2^\pht{t}
				\psi_t
			\right\rangle\right)
		\right|
		&\le\frac{
			\left\|
				v
			\right\|_{L^1\left(\mathbb R^3\right)}
		 }{ N^\frac{1+\lambda}2}
		\left\|
			\pht{t}
		\right\|_{L^\infty\left(\mathbb R^3\right)}^2
		\sqrt{\alpha_N^\lambda\left(\psi_t,\pht{t}\right)},
	\\
	\label{eq:twoqs}
		\left|
			\Im\left(\left\langle
				\psi_t,
				\left(
					\widehat\mu_2^{\lambda,\pht{t}}
					-\widehat\mu^{\lambda,\pht{t}}
				\right)
				p_1^\pht{t} p_2^\pht{t}
				U_{1,2}
				q_1^\pht{t} q_2^\pht{t}
				\psi_t
			\right\rangle\right)
		\right|&\le
	\\
		&\hspace{-150pt}\le 2
		\max\left\{
			\left\|v\right\|_{L^1\left(\mathbb R\right)},
			\left\|v\right\|_{L^2\left(\mathbb R\right)}
		\right\}
		\max\left\{
			\left\|\pht{t}\right\|_{L^4\left(\mathbb R^3\right)},
			\left\|\pht{t}\right\|_{L^\infty\left(\mathbb R^3\right)}
		\right\}^2
		\left(
			\alpha_N^\lambda\left(\psi_t,\pht{t}\right)
			+\frac{N^{3\beta-\lambda}}2
		\right),
	\\
	\label{eq:threeqs}
		\left|
			\Im\left(\left\langle
				\psi_t,
				\left(
					\widehat\mu_1^{\lambda,\pht{t}}
					-\widehat\mu^{\lambda,\pht{t}}
				\right)
				p_1^\pht{t} q_2^\pht{t}
				U_{1,2}
				q_1^\pht{t} q_2^\pht{t}
				\psi_t
			\right\rangle\right)
		\right|
		&\le
	\\
		&\hspace{-150pt}\le
		\frac{
			2\max\left\{
				\left\|v\right\|_{L^1\left(\mathbb R^3\right)},
				\left\|v\right\|_{L^2\left(\mathbb R^3\right)}
			\right\}
		 }{ N^\frac{1-\lambda}2}
		\left\|\pht{t}\right\|_{L^\infty\left(\mathbb R^3\right)}
		\left(
			N^{\frac{3\beta}2}
			+\left\|\pht{t}\right\|_{L^\infty\left(\mathbb R^3\right)}
		\right)
		\alpha_N^\lambda\left(\psi_t,\pht{t}\right).
	\end{align}

\end{pro}

\begin{proof}
In order to simplify the notation, during the proof we shall drop the labels $ \lambda $ and $ \pht{t} $ for short. 
	To start we note that for any function we can rewrite the operator $ p_1 v_N\left(x_1-x_2\right)p_1 $ in the following more convenient form:
	\begin{equation}\label{eq:p_convolves}
		p_1v_N\left(x_1-x_2\right)p_1=\left(
			|\pht{t}\rangle\langle\pht{t}|
		\right)_1
		v_N\left(x_1-x_2\right)
		\left(
			|\pht{t}\rangle\langle\pht{t}|
		\right)_1=\left(v_N*|\pht{t}|^2\right)\left(x_2\right) p_1.
	\end{equation}

	\noindent \underline{Proof of \eqref{eq:oneq}.} Using the fact $ p_jq_j=q_jp_j=0 $ for any $ 1\le j\le N $, by \eqref{eq:p_convolves}, we can rewrite 
	\begin{eqnarray}
		p_1p_2 U_{1,2}p_1q_2
		&=&p_1p_2\
		\left[
			\left(N-1\right)v_N\left(x_1-x_2\right)
			-N v_N*|\pht{t}|^2\left(x_1\right)
			-N v_N*|\pht{t}|^2\left(x_2\right)
		\right]
		p_1q_2
	\\
		&&-Np_1p_2\
		v_N*|\pht{t}|^2\left(x_2\right)
		p_1q_2\nonumber
	\\
		&=&-p_1p_2\
		v_N*|\pht{t}|^2\left(x_2\right)
		p_1q_2.\nonumber
	\end{eqnarray}
	As a consequence, using \eqref{itm:combinatorics} and \eqref{itm:commuting_hats}, we are able to bound the left-hand term in \eqref{eq:oneq} as
	\begin{align}
		\left|\Im\left(\left\langle
				\psi_t,
				\left(
					\widehat\mu_1
					-\widehat\mu
				\right)
				p_1p_2\
				v_N*|\pht{t}|^2\left(x_2\right)
				p_1q_2
				\psi_t
			\right\rangle\right)
		\right|
		&\le\left\|
			v_N*|\pht{t}|^2\left(x_2\right)
			p_2
		\right\|
		\left\|
			\left(
				\widehat\mu
				-\widehat\mu_{-1}
			\right)
			q_2
			\psi_t
		\right\|\nonumber
	\\
		&\le\left\|
			v_N*|\pht{t}|^2\left(x_2\right)
			p_2
		\right\|
		\left\|
			\left(
				\widehat\mu
				-\widehat\mu_{-1}
			\right)
			\widehat\nu
			\psi_t
		\right\|.
	\end{align}
We then bound separately these last two terms. To bound the first term, one can easily show that
		\begin{eqnarray}\label{eq: est 1}
			\left\|
				v_N*|\pht{t}|^2\left(x_2\right)
				p_2
			\right\|^2
			&\le&\langle
				\pht{t},
				\left(v_N*|\pht{t}|^2\left(x_2\right)\right)^2
				\pht{t}
			\rangle\le\left\|
				\pht{t}
			\right\|_{L^6\left(\mathbb R^3\right)}^2
			\left\|
				v_N*|\pht{t}|^2\left(x_2\right)
			\right\|_{L^3\left(\mathbb R^3\right)}^2
			\\
			&\le&\left\|
				\pht{t}
			\right\|_{L^6\left(\mathbb R^3\right)}^6
			\left\|
				v
			\right\|_{L^1\left(\mathbb R^3\right)}^2 \leq \left\|
			\pht{t}
		\right\|_{L^\infty\left(\mathbb R^3\right)}^{4}
		\left\|
			v
		\right\|_{L^1\left(\mathbb R^3\right)}^{2}.\nonumber 
		\end{eqnarray}
	
		To bound the second term we want to use the explicit definition of $ \mu $. In particular, we have that $|\mu(k) - \mu_{-1}(k)| \leq CN^{-\lambda} $. Moreover, for any $ k\le N^\lambda $, we have
		\begin{align}
			\left|
				\left(\mu\left(k\right)-\mu_{-1}\left(k\right)\right)^2
				\nu^2\left(k\right)
			\right|
			\le \frac k{N^{1+2\lambda}}
			=\frac1{N^{1+\lambda}}\mu\left(k\right).
		\end{align}
		Given that $\mu(k) - \mu_{-1}(k)  = 0$ if $ k \geq N^\lambda + 1$, we get that 
		\begin{equation}
		\label{eq:simple_mu_bound}
			\left\|
				\left(
					\widehat\mu
					-\widehat\mu_{-1}
				\right)
				\widehat\nu
				\psi_t
			\right\|
			=\left\langle
				\psi_t,
				\left(
					\widehat\mu
					-\widehat\mu_{-1}
				\right)^2
				\widehat\nu^2
				\psi_t
			\right\rangle^\frac12
		\le\frac1{N^\frac{1+\lambda}2}
			\left\langle
				\psi_t,
				\widehat\mu
				\psi_t
			\right\rangle^\frac12
			=\frac1{N^\frac{1+\lambda}2}
			\sqrt{\alpha_N^\lambda\left(\psi_t,\pht{t}\right)}.
		\end{equation}
	The estimates \eqref{eq: est 1} and \eqref{eq:simple_mu_bound} imply \eqref{eq:oneq}.\\
	\noindent{\underline{Proof of \eqref{eq:twoqs}}.} To prove \eqref{eq:twoqs}, we first look directly at the operator $ p_1p_2U_{1,2}q_1q_2 $. We have 
	\begin{equation}
		p_1p_2U_{1,2}q_1q_2
		=\left(N-1\right)p_1p_2v_N\left(x_1-x_2\right)q_1q_2.
	\end{equation}
We act similarly as before. In order to do so, notice first that the function $ \mu_2-\mu $ is positive and corresponds still to a weight, so we can consider the square root of the operator $ \widehat\mu_2-\widehat\mu $ and using \eqref{itm:commuting_hats}, we get
	\begin{multline}
	\left|
	\Im\left(\left\langle
		\psi_t,
		\left(
			\widehat\mu_2^{\lambda,\pht{t}}
			-\widehat\mu^{\lambda,\pht{t}}
		\right)
		p_1^\pht{t} p_2^\pht{t}
		U_{1,2}
		q_1^\pht{t} q_2^\pht{t}
		\psi_t
	\right\rangle\right)
\right|
		\\
		\le\left(N-1\right)
		\left|
			\left\langle
				\psi_t,
				\left(\widehat\mu_2-\widehat\mu\right)^\frac12
				p_1p_2
				v_N\left(x_1-x_2\right)
				q_1q_2
				\left(\widehat\mu-\widehat\mu_{-2}\right)^\frac12
				\psi_t
			\right\rangle
		\right|.\nonumber
	\end{multline}
	We now use that $ v_N\left(x_1-x_2\right) $ is nonzero only in a small region where $ x_1\approx x_2 $, because of the quick decay of $ v_N $. 
	To exploit this fact is then convenient to symmetrize $ \left(N-1\right)v_N\left(x_1-x_2\right) $ replacing it with $ \sum_{k=2}^Nv_N\left(x_1-x_k\right) $ to get, using Cauchy-Schwarz, 
	\begin{multline}\label{eq: start}
		\left(N-1\right)
		\left|
			\left\langle
				\psi_t,
				\left(\widehat\mu_2-\widehat\mu\right)^\frac12
				p_1p_2
				v_N\left(x_1-x_2\right)
				q_1q_2
				\left(\widehat\mu-\widehat\mu_{-2}\right)^\frac12
				\psi_t
			\right\rangle
		\right|
	\\ \le\left\|
			\sum_{j=2}^N
			q_j
			v_N\left(x_1-x_j\right)
			p_1p_j
			\left(\widehat\mu_2-\widehat\mu\right)^{1/2}\psi_t
		\right\|
		\left\|
			\left(\widehat\mu-\widehat\mu_{-2}\right)^{1/2}q_1\psi_t
		\right\|.
	\end{multline}
	We now bound those terms separately.
For the second term, we can again write explicitly the difference $ \widehat\mu-\widehat\mu_{-2} $: 
		\begin{equation}
			\widehat\mu\left(k\right)
			-\widehat\mu_{-2}\left(k\right)
			=\left\{
				\begin{array}{ll}
					\frac2{N^\lambda} & \mbox{for }k\le N^\lambda, \\
					\frac{N^\lambda+2-k }{ N^\lambda} & \mbox{for }N^\lambda<k\le N^\lambda+2, \\
					0 & \mbox{otherwise,}
				\end{array}
			\right.
		\end{equation}
and by \eqref{itm:combinatorics} we then get
		\begin{equation}
		\label{eq:second_derivative_mu}
			\left\|
				\left(\widehat\mu-\widehat\mu_{-2}\right)^{1/2}q_1\psi_t
			\right\|
			=\left\langle
				\psi_t,
				\left(\widehat\mu-\widehat\mu_{-2}\right)
				\widehat\nu^2
				\psi_t
			\right\rangle^\frac12
		\le\sqrt\frac2N
			\left\langle
				\psi_t,
				\widehat\mu
				\psi_t
			\right\rangle^\frac12
			=\sqrt\frac2N\sqrt{\alpha_N^\lambda\left(\psi_t,\pht{t}\right)}.
		\end{equation}
We now take into account the first term in the r.h.s of \eqref{eq: start}. Explicitly, it is 
		\begin{multline}\label{eq: eq 1a}
			\left\|\sum_{j=2}^Nq_jv_N\left(x_1-x_j\right)p_1p_j\left(\widehat\mu_2-\widehat\mu\right)^{1/2}\psi_t	\right\|^2
			\\
			=\sum_{j=2}^N\sum_{\ell=2}^N\left\langle\psi_t,\left(\widehat\mu_2-\widehat\mu\right)^{1/2}p_1p_jv_N\left(x_1-x_j\right)q_jq_\ell v_N\left(x_1-x_\ell\right)p_1p_\ell\left(\widehat\mu_2-\widehat\mu\right)^{1/2}\psi_t\right\rangle, 
		\end{multline}
	and now we bound it as a sum of two different terms, one in which $j= \ell$ and one for the case $j\neq \ell$.
		For the case $j \neq \ell$, by \eqref{itm:onep} and \eqref{itm:twops}, we can write
		\begin{multline}
			N^2\left\langle
				\psi_t,
				\left(\widehat\mu_2-\widehat\mu\right)^{1/2}
				p_1p_2
				v_N\left(x_1-x_2\right)
				q_2
				q_3
				v_N\left(x_1-x_3\right)
				p_1p_3
				\left(\widehat\mu_2-\widehat\mu\right)^{1/2}
				\psi_t
			\right\rangle
		\\
		\le N^2\left\|
				\sqrt{\left|
					v_N\left(x_1-x_2\right)
				\right|}
				\sqrt{\left|
					v_N\left(x_1-x_3\right)
				\right|}
				p_1p_3
				\left(\widehat\mu_2-\widehat\mu\right)^{1/2}
				q_2
				\psi_t
			\right\|^2\nonumber
		\\
		\le N^2
			\left\|
				\sqrt{\left|
					v_N\left(x_1-x_2\right)
				\right|}
				p_2
			\right\|^4
			\left\|
				\left(\widehat\mu_2-\widehat\mu\right)^{1/2}
				q_2
				\psi_t
			\right\|^2.
		\end{multline}
Moreover, by \eqref{itm:twops}, we get
		\begin{equation}
			\left\|
				\sqrt{\left|
					v_N\left(x_1-x_2\right)
				\right|}
				p_2
			\right\|^4
			=\left\|
				p_2
				\left|
					v_N\left(x_1-x_2\right)
				\right|
				p_2
			\right\|^2
		\leq \left\|\pht{t}\right\|_{L^\infty\left(\mathbb R^3\right)}^4
			\left\|v\right\|_{L^1\left(\mathbb R^3\right)}^2,\nonumber
		\end{equation}
		and on the other hand proceeding similarly as in \eqref{eq:second_derivative_mu} we can estimate
		\begin{equation}
			\left\|
				\left(\widehat\mu_2-\widehat\mu\right)^{1/2}
				q_2
				\psi_t
			\right\|^2
			\le\frac2N\alpha_N^\lambda\left(\psi_t,\pht{t}\right).
		\end{equation}
		We then get
		\begin{multline}\label{eq: eq 1}
			\left|N^2\left\langle
				\psi_t,
				\left(\widehat\mu_2-\widehat\mu\right)^{1/2}
				p_1p_2
				v_N\left(x_1-x_2\right)
				q_2
				q_3
				v_N\left(x_1-x_3\right)
				p_1p_3
				\left(\widehat\mu_2-\widehat\mu\right)^{1/2}
				\psi_t
			\right\rangle\right|
			\\
		\le2N\left\|v\right\|_{L^1\left(\mathbb R^3\right)}^2
			\left\|\pht{t}\right\|_{L^\infty\left(\mathbb R^3\right)}^4
			\alpha_N^\lambda\left(\psi_t,\pht{t}\right).
		\end{multline}
We then look at the second term, i.e., the case $j=\ell$, in the right hand side of \eqref{eq: eq 1a}: for this term we cannot bring any $ q $ close to the state $ \psi_t $, and is therefore not possible to exploit the state $ \psi_t $. We then bound every operator in it directly in norm using \eqref{itm:fourps}, we get 
		\begin{multline}\label{eq: eq 2}
			N\left\langle
				\psi_t,
				\left(\widehat\mu_2-\widehat\mu\right)^{1/2}
				p_1p_2
				v_N\left(x_1-x_2\right)
				q_2
				v_N\left(x_1-x_2\right)
				p_1p_2
				\left(\widehat\mu_2-\widehat\mu\right)^{1/2}
				\psi_t
			\right\rangle
		\\
			\le N\left\|
				v_N\left(x_1-x_2\right)p_1p_2
			\right\|^2
			\left\|
				\left(
					\widehat\mu_2
					-\widehat\mu
				\right)^\frac12
			\right\|^2
			\leq 2N^{1+3\beta-\lambda}
			\left\|\pht{t}\right\|_{L^4\left(\mathbb R^3\right)}^4
			\left\|v\right\|_{L^2\left(\mathbb R^3\right)}^2.
		\end{multline}
Combining \eqref{eq: eq 1} and \eqref{eq: eq 2}, we can conclude that 
		\begin{multline}\label{eq: eq 3}
			\left\|
				\sum_{j=2}^N
				q_j
				v_N\left(x_1-x_j\right)
				p_1p_j
				\left(\widehat\mu_2-\widehat\mu\right)^{1/2}\psi_t
			\right\|
		\\
			\qquad
			\le\sqrt{
				2N\left\|v\right\|_{L^1\left(\mathbb R\right)}^2
				\left\|\pht{t}\right\|_{L^\infty\left(\mathbb R^3\right)}^4
				\alpha_N^\lambda\left(\psi_t,\pht{t}\right)
				+2N^{1+3\beta-\lambda}
				\left\|\pht{t}\right\|_{L^4\left(\mathbb R^3\right)}^4
				\left\|v\right\|_{L^2\left(\mathbb R^3\right)}^2.
			}
		\end{multline}
Moreover, the estimates \eqref{eq:second_derivative_mu} and \eqref{eq: eq 3}, imply \eqref{eq:twoqs}.\\
	\noindent\underline{Proof of \eqref{eq:threeqs}
	}. 
	We start the proof of \eqref{eq:threeqs} by writing 
	\begin{align}
	\nonumber
	\left|
	\Im\left(\left\langle
		\psi_t,
		\left(
			\widehat\mu_1^{\lambda,\pht{t}}
			-\widehat\mu^{\lambda,\pht{t}}
		\right)
		p_1^\pht{t} q_2^\pht{t}
		U_{1,2}
		q_1^\pht{t} q_2^\pht{t}
		\psi_t
	\right\rangle\right)
\right|
	&\le\left\|
			\left(\widehat\mu_1-\widehat\mu\right)^\frac12
			q_2\psi_t
		\right\|
		\left\| U_{1,2}p_1\right\|
		\left\|
			\left(\widehat\mu-\widehat\mu_{-1}\right)^{1/2}
			q_1q_2
			\psi_t
		\right\|\nonumber
	\\
	\label{eq:third_term_almost_finished}
		&\le\sqrt{\frac N{N-1}}
		\left\|
			\left(\widehat\mu_1-\widehat\mu\right)^{\frac{1}{2}}\widehat\nu\psi_t
		\right\|
		\left\| U_{1,2}p_1\right\|
		\left\|
			\left(\widehat\mu-\widehat\mu_{-1}\right)^{\frac{1}{2}}\widehat\nu^2
			\psi_t
		\right\|.
	\end{align}
	Proceeding as above we can estimate directly $ \left(\widehat\mu_1-\widehat\mu\right)\widehat\nu^2\le\frac1N\widehat\mu $ and $ \left(\widehat\mu-\widehat\mu_{-1}\right)\widehat\nu^4\le\frac1{N^{2-\lambda}}\widehat\mu $.
	We now bound the potential term. We have, using \eqref{itm:onep}, that
	\begin{equation}\label{eq: bound U12p1}
		\left\|U_{1,2}p_1\right\|
		\le N\left\|\pht{t}\right\|_{L^\infty\left(\mathbb R^3\right)}
		\max\left\{
			\left\|v\right\|_{L^1\left(\mathbb R^3\right)},
			\left\|v\right\|_{L^2\left(\mathbb R^3\right)}
		\right\}
		\left(
			N^{\frac{3\beta}2}
			+\left\|\pht{t}\right\|_{L^\infty\left(\mathbb R^3\right)}
		\right).
	\end{equation}
	Inserting \eqref{eq: bound U12p1} in \eqref{eq:third_term_almost_finished}, we get \eqref{eq:threeqs}.
\end{proof}

Combining Lemma \ref{lem:derivative_alpha} and \cref{pro:terms_of_alpha}, we can prove \cref{pro: Gronwall alpha}.

\begin{proof}[Proof of \cref{pro: Gronwall alpha}]
	
	As stated before, the key idea of the proof is to exploit a Grönwall-type argument. In order to do that, we use Lemma \ref{lem:derivative_alpha} and Proposition \ref{pro:terms_of_alpha} to write 
	\begin{align}
		\partial_t\alpha_N^\lambda\left(\psi_t,\pht{t}\right)
		&=\Gamma_N^\lambda\left(\psi_t,\pht{t}\right)\nonumber
	\\
		&\leq C_vg_N\left(\|\varphi_t^{\mathrm{H}}\|_{L^{\infty}(\mathbb{R}^3)}\left[1+ \frac{1 }{ N^{\frac{1-\lambda-3\beta}{2}}}\right] + \|\varphi_t^{\mathrm{H}}\|_{L^{\infty}(\mathbb{R}^3)}^2 \left[ 1 + \frac{1 }{ N^{\frac{1-\lambda }{ 2}}}\right] \right)\alpha^{\lambda}_N(\psi_y, \varphi^{\mathrm{H}}_t)\nonumber
	\\
		&\quad +C_vg_N
		\frac{\left\|
		\pht{t}
	\right\|_{L^\infty\left(\mathbb R^3\right)}^2}{ N^\frac{1+\lambda}2}
		\sqrt{\alpha_N^\lambda\left(\psi_t,\pht{t}\right)} +  \frac{C_vg_N }{ N^{\lambda-3\beta}} \left(		\left\|\pht{t}\right\|_{L^\infty\left(\mathbb R^3\right)} + 
			\left\|\pht{t}\right\|_{L^\infty\left(\mathbb R^3\right)}^2\right),
	\end{align}
	where $ C_v $ is a constant that depends only on the $ L^1 $ and $ L^2 $ norms of $ v $.
We can now do Cauchy-Schwarz in the next to the last term and discard some terms which are subleading for any $\lambda\in (0, 1- 3\beta)$, to get
\begin{eqnarray}
	\partial_t\alpha_N^\lambda\left(\psi_t,\pht{t}\right) &\leq& C_v g_N\left(\|\varphi_t^{\mathrm{H}}\|_{L^{\infty}(\mathbb{R}^3)} + \|\varphi_t^{\mathrm{H}}\|_{L^{\infty}(\mathbb{R}^3)}^2\right)\alpha^{\lambda}_N(\psi_y, \varphi^{\mathrm{H}}_t) \nonumber
	\\
	&&+ \frac{C_v g_N }{ N^{1+\lambda}} \|\varphi_t^{\mathrm{H}}\|_{L^\infty(\mathbb{R}^3)}^2 + \frac{C_v g_N }{ N^{\lambda-3\beta}} \left(		\left\|\pht{t}\right\|_{L^\infty\left(\mathbb R^3\right)} + 
	\left\|\pht{t}\right\|_{L^\infty\left(\mathbb R^3\right)}^2\right).
\end{eqnarray}
Being $\lambda \in (3\beta, 1-3\beta)$, we get
\begin{equation}
	\partial_t\alpha_N^\lambda\left(\psi_t,\pht{t}\right) \leq C_v g_N\left(\|\varphi_t^{\mathrm{H}}\|_{L^{\infty}(\mathbb{R}^3)} + \|\varphi_t^{\mathrm{H}}\|_{L^{\infty}(\mathbb{R}^3)}^2\right)\left( \alpha^{\lambda}_N(\psi_y, \varphi^{\mathrm{H}}_t) + \frac{1 }{ N^{\lambda -3\beta}} \right)
\end{equation}
To bound $\|\varphi^{\mathrm{H}}_t\|_{L^\infty(\mathbb{R}^3)}$, we use \cref{pro: H2 norm}. We find
\begin{equation}\label{eq: def C(phi0)}
	\|\varphi^{\mathrm{H}}_t\|_{L^\infty(\mathbb{R}^3)} \leq \|\varphi^{\mathrm{H}}_t\|_{H^2(\mathbb{R}^3)} 
	\leq\left\|
		\varphi_0
	\right\|_{H^2\left(\mathbb R^3\right)}
	e^{Cg_N^2\left([\mathcal{E}^{\mathrm{free}}_{\mathrm{GP}}(\varphi_0)]^2 + g_N^2N^{-2\beta}\|\varphi_0\|_{L^\infty(\mathbb{R}^3)}^4\right)\left|t\right|} = C_N(\varphi_0,t).
\end{equation}

	By Grönwall, we can then conclude that 
	\begin{equation}
		\alpha_N^\lambda\left(\psi_t,\pht{t}\right)
		\le\left(
			\alpha_N^\lambda\left(\psi_0,\varphi_0\right)
			+\frac1{N^{\lambda-3\beta}}
		\right)
		e^{C_{v}(C_N(\varphi_0,t) + C_N(\varphi_0,t)^2)g_N\left|t\right|}.
	\end{equation}
	
\end{proof}

\subsubsection{Proof of \cref{pro: from MB to H}} The proof of Proposition \ref{pro: from MB to H} is split in two parts:
\begin{itemize}
\item First, we prove that 
\begin{equation}\label{eq: second step}
	\left\||\varphi^{\mathrm{H}}_t\rangle \langle \varphi^{\mathrm{H}}_t| - \gamma_{\psi_t}^{(1)}\right\|^2 \leq 2\alpha^{\lambda}_N(\psi_t, \varphi^{\mathrm{H}}_t).
\end{equation}
\item Secondly, we prove that for some $\lambda\in (0,1)$, for any $\psi\in L^{2}_s(\mathbb{R}^{3N})$ with $\| \psi\|_{ L^{2}_s(\mathbb{R}^{3N})}=1$ and for any normalized $\varphi\in L^2(\mathbb{R}^3)$, one has
\begin{equation}\label{eq: first step}
	\alpha^{\lambda}_N(\psi, \varphi) \leq N^{1-\lambda}\left\| |\varphi\rangle \langle\varphi| -\gamma_{\psi}^{(1)}\right\|.
\end{equation}
\end{itemize}
Combining the estimates in \eqref{eq: second step} and in \eqref{eq: first step} with the result in \cref{pro: Gronwall alpha} we get the result.
\begin{proof}[Proof of \cref{pro: from MB to H}]
We start by proving \eqref{eq: second step}. Being $|\pht{t}\rangle\langle\pht{t}|-\gamma_{\psi_t}^{\left(1\right)}\in \mathfrak S_2\left(L^2\left(\mathbb R^3\right)\right) $, where with $\mathfrak S_2$ we denote the $p$ Schatten space with $p=2$, we have, 
		\begin{align}
			\left\|
				|\pht{t}\rangle\langle\pht{t}|-\gamma_{\psi_t}^{\left(1\right)}
			\right\|
			\le\left\|
				|\pht{t}\rangle\langle\pht{t}|-\gamma_{\psi_t}^{\left(1\right)}
			\right\|_{\mathfrak S_2\left(L^2\left(\mathbb R^3\right)\right)}.
		\end{align}
	Moreover, using \eqref{itm:combinatorics}, we can write 
		\begin{eqnarray}
			\left\|
				|\pht{t}\rangle\langle\pht{t}|-\gamma_{\psi_t}^{\left(1\right)}
			\right\|_{\mathfrak S_2\left(L^2\left(\mathbb R^3\right)\right)}^2
			&\le &2\left(
				1
				-\tr\left[
					|\pht{t}\rangle\langle\pht{t}|
					\gamma_{\psi_t}^{\left(1\right)}
				\right]
			\right)
			=2\left(
				1
				-\left\|p_1^\pht{t}\psi_t\right\|_{L_\mathrm s^2\left(\mathbb R^{3N}\right)}^2
			\right)
			\\
			&=&2\left\|q_1^\pht{t}\psi_t\right\|_{L_\mathrm s^2\left(\mathbb R^{3N}\right)}^2
			=2\left\|\widehat\nu^\pht{t}\psi_t\right\|_{L_\mathrm s^2\left(\mathbb R^{3N}\right)}^2.\nonumber
		\end{eqnarray}
	Given that $ \nu^2\le\mu^\lambda $ we can write
	\begin{equation}\label{eq: est proj MB to H}
		\left\|
			|\pht{t}\rangle\langle\pht{t}|-\gamma_{\psi_t}^{\left(1\right)}
		\right\|
		\le\sqrt2\left\|\widehat\nu^\pht{t}\psi_t\right\|_{L_\mathrm s^2\left(\mathbb R^{3N}\right)}
	\le\sqrt{2\alpha_N^\lambda\left(\psi_t,\pht{t}\right)}.
	\end{equation}
Consider now $ \alpha_N^\lambda\left(\psi,\varphi\right) = \langle \psi \hat{\mu}^{\lambda, \varphi}\psi\rangle  $, we have 
\begin{equation}
	\alpha_N^\lambda\left(\psi,\varphi\right)
	\le N^{1-\lambda}\left\langle
		\psi,
		\left(\widehat\nu^\varphi\right)^2
		\psi
	\right\rangle
	=N^{1-\lambda}\left\|
		q_1^\varphi\psi
	\right\|_{L_\mathrm s^2\left(\mathbb R^{3N}\right)}^2
	\le N^{1-\lambda}\left\|
		|\varphi\rangle\langle\varphi|
		-\gamma_\psi^{\left(1\right)}
	\right\|,
\end{equation}
where we used that $\mu^\lambda (k) \leq N^{1-\lambda}\nu^2$ and where the last inequality follows from the definition of $q_1^\varphi$. If we now apply \eqref{eq: first step} to $\psi_0\in L^{2}(\mathbb{R}^{3N})$ and $\varphi_0\in L^2(\mathbb{R}^3)$, by \eqref{eq: second step} and \cref{pro: Gronwall alpha}, we get
\begin{align}\label{eq: est alpha MB to H}
	\left\||\varphi^{\mathrm{H}}_t\rangle \langle \varphi^{\mathrm{H}}_t| - \gamma_{\psi_t}^{(1)}\right\|^2
	\le 2\left(
		N^{1-\lambda}\left\||\varphi_0\rangle\langle\varphi_0|-\gamma_{\psi_0}^{\left(1\right)}\right\|
		+\frac1{N^{\lambda-3\beta}}
	\right)e^{C_{v}(C_N(\varphi_0,t) + C_N(\varphi_0,t)^2)g_N\left|t\right|}.
\end{align}
\end{proof} 

\subsection{From Hartree to Gross-Pitaevskii}
The main goal of this section is to prove the following proposition.
\begin{pro}[From Hartree to Gross-Pitaevskii]\label{pro: from H to GP} Let $\beta\in (0,1/3)$. Let $ \varphi_t^{\mathrm{GP}} $ and  $ \pht{t} $ be the normalized solution of \eqref{eq:GP} and \eqref{eq:Hartree}, respectively, both with initial data  $ \varphi_0\in L^2(\mathbb{R}^3) $
Suppose that $ v $ satisfies Assumption \ref{asump: 1}. We have
\begin{equation}\label{eq: from H to GP}
	\|\varphi^{\mathrm{GP}}_t - \varphi^{\mathrm{H}}_t\|_{L^2(\mathbb{R}^3)}
	\leq C\sqrt{g_N}\frac{1+ [\mathcal{E}^{\mathrm{free}}_{\mathrm{GP}}(\varphi_0)]^2 + g_NN^{-\beta}\|\varphi_0\|_{L^\infty}^2 }{ N^\frac\beta2}e^{C_v C_N(\varphi_0,t)^2g_N|t|},
\end{equation}
where $C_v$ is a constant which depends only on $v$ and $C_N(\varphi_0,t)$ is given by 
\begin{equation}
	C_N(\varphi_0,t)
	=\left\|
		\varphi_0
	\right\|_{H^2\left(\mathbb R^3\right)}
	e^{Cg_N^2\left([\mathcal{E}^{\mathrm{free}}_{\mathrm{GP}}(\varphi_0)]^2 + g_N^2N^{-2\beta}\|\varphi_0\|_{L^\infty(\mathbb{R}^3)}^4\right)\left|t\right|}
\end{equation}

\end{pro}
\begin{proof} We prove \eqref{eq: from H to GP} using Grönwall estimate. To do that we want to bound the following quantity
	\begin{eqnarray}
		\partial_t\|\varphi^{\mathrm{GP}}_t - \varphi^{\mathrm{H}}_t\|_{L^2(\mathbb{R}^3)}^2 = 2\mathrm{Re}\,\langle \varphi^{\mathrm{GP}}_t - \varphi^{\mathrm{H}}_t, \partial_t(\varphi^{\mathrm{GP}}_t - \varphi^{\mathrm{H}}_t)\rangle 
		&=&  2g_N(\smallint v)\Im\, \langle \varphi^{\mathrm{H}}_t, (|\varphi^{\mathrm{H}}_t|^2 - |\varphi^{\mathrm{GP}}_t|^2)\varphi_t^{\mathrm{GP}}\rangle \nonumber\\
		&& +\, 2g_N(\smallint v)\Im\, \langle \varphi^{\mathrm{GP}}_t, (|\varphi^{\mathrm{H}}_t|^2 - v_N\ast|\varphi^{\mathrm{H}}_t|^2)\varphi^{\mathrm{H}}_t\rangle\nonumber.
	\end{eqnarray}
We now estimate the two quantities above separately. For the first one, we have
\begin{eqnarray}
	\left |  2g_N\Im\, \langle \varphi^{\mathrm{H}}_t, (|\varphi^{\mathrm{H}}_t|^2-|\varphi^{\mathrm{GP}}_t|^2 )\varphi_t^{\mathrm{GP}}\rangle  \right| &\leq& C g_N \|\varphi^{\mathrm{GP}}_t\|_{L^\infty(\mathbb{R}^3)} (\|\varphi^{\mathrm{GP}}_t\|_{L^\infty(\mathbb{R}^3)} + \|\varphi^{\mathrm{H}}_t\|_{L^\infty(\mathbb{R}^3)}) \|\varphi^{\mathrm{GP}}_t - \varphi^{\mathrm{H}}_t\|_{L^2(\mathbb{R}^3)}^2\nonumber
	\\
	&\leq& C g_N (\|\varphi^{\mathrm{GP}}_t\|_{L^\infty(\mathbb{R}^3)}^2 + \|\varphi^{\mathrm{H}}_t\|_{L^\infty(\mathbb{R}^3)}^2) \|\varphi^{\mathrm{GP}}_t - \varphi^{\mathrm{H}}_t\|_{L^2(\mathbb{R}^3)}^2.
\end{eqnarray}

We now estimate the second one, we have 
\begin{eqnarray}
	| 2g_N \Im\,\hspace{-0.7cm}&&\langle \varphi^{\mathrm{GP}}_t, (|\varphi^{\mathrm{H}}_t|^2 - v_N\ast|\varphi^{\mathrm{H}}_t|^2)\varphi^{\mathrm{H}}_t\rangle|
	\\
	&&\leq 2g_N \int_{\mathbb{R}^3}dx\, |\varphi^{\mathrm{GP}}_t(x)||\varphi^{\mathrm{H}}_t(x)|\int_{\mathbb{R}^3}dy\, v(y) \left(|\varphi^{\mathrm{H}}_t(x)|^2 - \left|\varphi^{\mathrm{H}}_t\left(x-\frac{y }{ N^\beta}\right)\right|^2\right)\nonumber
	\\
	&&= 2g_N \int_{\mathbb{R}^3}dx\int_{\mathbb{R}^3}dy\, v(y)|\varphi^{\mathrm{GP}}_t(x)||\varphi^{\mathrm{H}}_t(x)|\int_0^1 d\lambda\frac{d }{ d\lambda}\left|\varphi^{\mathrm{H}}_t\left(x-\frac{\lambda y }{ N^\beta}\right)\right|^2 \nonumber
	\\
	&&\leq 4\frac{g_N }{ N^\beta} \int_{\mathbb{R}^3}dx\int_{\mathbb{R}^3}dy\,\int_0^1 d\lambda\,|y|v(y)|\varphi^{\mathrm{GP}}_t(x)||\varphi^{\mathrm{H}}_t(x)|\left|\varphi^{\mathrm{H}}_t\left(x -\frac{\lambda y }{ N^\beta}\right)\right|\left|\nabla\varphi^{\mathrm{H}}_t\left(x -\frac{\lambda y }{ N^\beta}\right)\right|\nonumber
	\\
	&&\leq C\frac{g_N }{ N^\beta} \|xv\|_{L^1(\mathbb{R}^3)}\|\varphi^{\mathrm{H}}_t\|_{L^\infty(\mathbb{R}^3)}^2 (1+\|\nabla\varphi^{\mathrm{H}}_t\|_{L^2(\mathbb{R}^3)}^2).
\end{eqnarray}
Then we find
\begin{equation}
	\partial_t\|\varphi^{\mathrm{GP}}_t - \varphi^{\mathrm{H}}_t\|_{L^2(\mathbb{R}^3)}^2 \leq  Cg_N \left(\|\varphi^{\mathrm{GP}}_t\|_{L^\infty(\mathbb{R}^3)}^2 + \|\varphi^{\mathrm{H}}_t\|_{L^\infty(\mathbb{R}^3)}^2\right)\,\left( \|\varphi^{\mathrm{GP}}_t - \varphi^{\mathrm{H}}_t\|_{L^2(\mathbb{R}^3)}^2  + \frac{1+\|\nabla\varphi^{\mathrm{H}}_t\|_{L^2(\mathbb{R}^3)}^2 }{ N^\beta}\right)
\end{equation}

We can bound for any time $t$ the $L^\infty$ norms of $\varphi^{\mathrm{GP}}_t$, $\varphi^{\mathrm{H}}_t$, using the bound for the $H^2$ norms proved in \cref{pro: H2 norm}, i.e., 
\begin{equation}
	\|\varphi^{\mathrm{GP}}_t\|_{L^\infty(\mathbb{R}^3)} \leq \|\varphi_0\|_{H^2(\mathbb{R}^3)} e^{Cg_N^2[\mathcal{E}^{\mathrm{free}}_{\mathrm{GP}}(\varphi_0)]^2 \left|t\right|},
\end{equation}
\begin{equation}
	\|\varphi^{\mathrm{H}}_t\|_{L^\infty(\mathbb{R}^3)} 
	\leq C\left\|
		\varphi_0
	\right\|_{H^2\left(\mathbb R^3\right)}
	e^{Cg_N^2\left([\mathcal{E}^{\mathrm{free}}_{\mathrm{GP}}(\varphi_0)]^2 + g_N^2N^{-2\beta}\|\varphi_0\|_{L^\infty(\mathbb{R}^3)}^4\right)\left|t\right|}.
\end{equation}
Moreover,by \cref{pro: H2 norm}, we also can bound uniformly in time
\begin{equation}
	\|\nabla\varphi^{\mathrm{H}}_t\|_{L^2(\mathbb{R}^3)} \leq  C\mathcal{E}^{\mathrm{free}}_{\mathrm{GP}}(\varphi_0) + Cg_NN^{-\beta}\|\varphi_0\|_{L^\infty}^2.
\end{equation}
Thus, by Grönwall, we find
\begin{equation}
 \|\varphi^{\mathrm{GP}}_t - \varphi^{\mathrm{H}}_t\|_{L^2(\mathbb{R}^3)}
 \leq C\sqrt{g_N}\frac{1+ \mathcal{E}^{\mathrm{free}}_{\mathrm{GP}}(\varphi_0) + g_NN^{-\beta}\|\varphi_0\|_{L^\infty}^2 }{ N^\frac\beta2}e^{C_v C_N(\varphi_0,t)^2g_N|t|},
\end{equation}
where $C_v$ is a constant which depends on $\|xv\|_{L^1(\mathbb{R}^3)}$ and $C(\varphi_0)$ is defined as in \eqref{eq: def C(phi0)}, i.e., 
\begin{equation}
	C_N(\varphi_0,t)
	=\left\|
		\varphi_0
	\right\|_{H^2\left(\mathbb R^3\right)}
	e^{Cg_N^2\left([\mathcal{E}^{\mathrm{free}}_{\mathrm{GP}}(\varphi_0)]^2+ g_N^2N^{-2\beta}\|\varphi_0\|_{L^\infty(\mathbb{R}^3)}^4\right)\left|t\right|}.
\end{equation}
\end{proof}

\subsection{Proof of Theorem \ref{thm:GP_cond_main}}

In this last section we combine the results of Proposition \ref{pro: Gronwall alpha} and Proposition \ref{pro: from H to GP} to prove \eqref{eq: dynamics}. We start by writing
\begin{equation}
	\left\|
		\gamma_{\psi_t}^{(1)}
		-P_\pgpt{t}
	\right\|
	\le\left\|
		\gamma_{\psi_t}^{(1)}
		-P_\pht{t}
	\right\|
	+\left\|
		P_\pht{t}
		-P_\pgpt{t}
	\right\|\le\left\|
		\gamma_{\psi_t}^{(1)}
		-P_\pht{t}
	\right\|
	+2\left\|
		\pht{t}
		-\pgpt{t}
	\right\|_{L^2\left(\mathbb{R}^3\right)}.
\end{equation}

By \cref{pro: from MB to H}, we get
\begin{align}
	\left\|
		\gamma_{\psi_t}^{(1)}
		-P_\pht{t}
	\right\|
	&\le\sqrt2\left(
		N^\frac{1-\lambda}2
		\left\|
			\gamma_{\psi_0}^{(1)}
			-P_{\varphi_0}
		\right\|^\frac12
		+\frac1{N^\frac{\lambda-3\beta}2}
	\right)
	e^{C_{v}(C_N(\varphi_0,t) + C_N(\varphi_0,t)^2)g_N\left|t\right|}.
\end{align}

On the other hand, from \cref{pro: from H to GP} we have
\begin{equation}
	2\left\|
		\pht{t}
		-\pgpt{t}
	\right\|_{L^2\left(\mathbb R^3\right)}
	\le C\sqrt{g_N}\frac{1+\mathcal{E}^{\mathrm{free}}_{\mathrm{GP}}(\varphi_0) + g_NN^{-\beta}\|\varphi_0\|_{L^\infty}^2 }{ N^\frac\beta2}e^{C_v C_N(\varphi_0,t)^2g_N|t|}.
\end{equation}

Combining the two previous results we then get
\begin{align}
	\left\|
		\gamma_{\psi_t}^{(1)}
		-P_\pgpt{t}
	\right\|
	&\le\sqrt2\left(
		N^\frac{1-\lambda}2
		\left\|
			\gamma_{\psi_0}^{(1)}
			-P_{\varphi_0}
		\right\|^\frac12
		+\frac1{N^\frac{\lambda-3\beta}2}
	\right)
	e^{C_{v}(C_N(\varphi_0,t) + C_N(\varphi_0,t)^2)g_N\left|t\right|}
\\
	&\qquad
	+C\sqrt{g_N}\frac{1+ \mathcal{E}^{\mathrm{free}}_{\mathrm{GP}}(\varphi_0) + g_NN^{-\beta}\|\varphi_0\|_{L^\infty}^2 }{ N^\frac\beta2}e^{C_v C_N(\varphi_0,t)^2g_N|t|},\nonumber
\end{align}
where $C_v$ now is a constant which depends on $\|v\|_{L^1(\mathbb{R}^3)}$, $\|v\|_{L^2(\mathbb{R}^3)}$ and $\| xv\|_{L^1(\mathbb{R}^3)}$.

\appendix

\section{Properties of the projectors $p^{\varphi}_j$, $q^{\varphi}_j$}\label{app: projector}
\begin{pro}
	\label{pro:properties_of_p}
	
		Let $ p_j^\varphi $ and $ q_j^\varphi $ be defined as in \eqref{eq: def pj qj}. 
		\begin{enumerate}

		\item[(i)]\label{itm:n_squared}
			Let $ \nu:\left\{0, 1,\ldots, N\right\}\to \mathbb R $ be given by $ \nu\left(k\right) :=\sqrt{\frac kN} $. Then the square of $ \widehat\nu^\varphi $ is the fraction of particles not in the state $ \varphi $, i.e.,
			\begin{equation}
				\left(
					\widehat{\nu}^\varphi
				\right)^2
				=\frac{1 }{ N}
				\sum_{j=1}^N
				q_j^\varphi.
			\end{equation}
	
		\item[(ii)]
			For any $ f:\left\{0, 1,\ldots, N\right\}\to \mathbb R $ and any symmetric state $ \psi \in L_\mathrm s^2\left(\mathbb R^{3N}\right) $ 
			\begin{equation}
				\label{itm:combinatorics}
				\left\|
					\widehat{f}^\varphi
					q^\varphi_1
					\psi
				\right\|_{L_\mathrm s^2\left(\mathbb R^3\right)}
				=
				\left\|
					\widehat{f}^\varphi \widehat{\nu}^\varphi \psi
				\right\|_{L_\mathrm s^2\left(\mathbb R^3\right)}, \quad \left\|
					\widehat{f}^\varphi
					q^\varphi_1
					q^\varphi_2
					\psi
				\right\|_{L_\mathrm s^2\left(\mathbb R^3\right)}
				\le
				\sqrt{\frac{N }{ N-1}}
				\left\|
					\widehat{f}^\varphi\left(\widehat{\nu}^\varphi\right)^2 \psi
				\right\|_{L_\mathrm s^2\left(\mathbb R^3\right)}.
			\end{equation}
	
		\item[(iii)]
			For any $ f:\left\{0, 1,\ldots, N\right\}\to \mathbb R $, $ v:\mathbb R^3\times\mathbb R^3\to \mathbb{R} $ and $ j, k =0, 1, 2 $, we have 
			\begin{equation}\label{itm:commuting_hats}
				\widehat{f}^\varphi
				Q^\varphi_j
				v\left(x_1,x_2\right)
				Q^\varphi_k=
				Q^\varphi_j
				v\left(x_1,x_2\right)
				Q^\varphi_k
				\widehat{f}^\varphi_{j-k},
			\end{equation}
			where $ Q^\varphi_0:=p^\varphi_1 p^\varphi_2 $, $ Q^\varphi_1:=p^\varphi_1 q^\varphi_2 $ and $ Q^\varphi_2:=q^\varphi_1 q^\varphi_2 $.
	
		\end{enumerate}
	
	\end{pro}

	\begin{proof}
	
	To prove \ref{itm:n_squared} note that $ \bigcup_{k=0}^N \mathcal{A}_k =\left\{0, 1\right\}^N $, so that $ 1 =\sum_{k=0}^N P_k^\varphi $. Using also $ \left(q_j^\varphi\right)^2=q_j^\varphi $ and $ q_j^\varphi p_j^\varphi =0 $, we get
	\begin{equation}
		q_j^\varphi
		\prod_{l=1}^N
		\left[
			\left(
				p_l^\varphi
			\right)^{1-a_l}
			\left(
				q_l^\varphi\
			\right)^{a_l}
		\right]
		=\left\{
			\begin{array}{ll}
				\prod_{l=1}^N
				\left[
					\left(
						p_l^\varphi
					\right)^{1-a_l}
					\left(
						q_l^\varphi\
					\right)^{a_l}
				\right] & \mbox{if }a_j=1, \\
				0. & \mbox{else}
			\end{array}
		\right\}
		=a_j
		\prod_{l=1}^N
		\left[
			\left(
				p_l^\varphi
			\right)^{1-a_l}
			\left(
				q_l^\varphi\
			\right)^{a_l}
		\right].
	\end{equation}
	
	Given that for any $ \mathbf a\in\mathcal A_k $ the sum of the $ a_j $'s is $ k $ we deduce that
	\begin{equation}
		\frac1N\sum_{j=1}^N q_j^\varphi
		=\frac1N\sum_{j=1}^N q_j^\varphi\sum_{k=0}^N P_k^\varphi
		=\frac1N\sum_{k=0}^N \sum_{j=1}^N q_j^\varphi P_k^\varphi
		=\frac1N\sum_{k=0}^N k P_k^\varphi
		=\sum_{k=0}^N\left(\nu\left(k\right)\right)^2P_k^\varphi
	\end{equation}
	and \ref{itm:n_squared} follows.

	We now prove \eqref{itm:combinatorics}. We can use the symmetry of $ \psi $ and the previous point to get
	\begin{equation}
		\left\|
			\widehat{f}^\varphi
			\widehat{\nu}^\varphi
			\psi
		\right\|_{L_\mathrm s^2\left(\mathbb R^3\right)}^2
		=\left\langle
			\psi, 
			\left(\widehat{f}^\varphi\right)^{2}
			\left(\widehat{\nu}^\varphi\right)^2
			\psi
		\right\rangle
		=\left\langle
			\psi, 
			\left(\widehat{f}^\varphi\right)^{2}
			q_1^\varphi
			\psi
		\right\rangle
		=\left\|
			\widehat{f}^\varphi
			q^\varphi_1
			\psi
		\right\|_{L_\mathrm s^2\left(\mathbb R^3\right)}^2.
	\end{equation}
	
	Similarly, for the second inequality in \eqref{itm:combinatorics}, we have 
	\begin{equation}
		\left\|
			\widehat{f}^\varphi
			\left(\widehat{\nu}^\varphi\right)^2
			\psi
		\right\|_{L_\mathrm s^2\left(\mathbb R^3\right)}^2
		=\frac1{N^2}\sum_{j, k=1}^N
		\left\langle
			\psi, 
			\left(\widehat{f}^\varphi\right)^2
			q_j^\varphi
			q_k^\varphi
			\psi
		\right\rangle
	=\frac{N-1 }{ N}
		\left\|
			\widehat{f}^\varphi
			q^\varphi_1
			q^\varphi_2
			\psi
		\right\|_{L_\mathrm s^2\left(\mathbb R^3\right)}^2
		+\frac1N\left\|
			\widehat{f}^\varphi
			\widehat{\nu}^\varphi
			\psi
		\right\|_{L_\mathrm s^2\left(\mathbb R^3\right)}^2.
	\end{equation}
	
	To prove \eqref{itm:commuting_hats} we use the definition of the $ Q_j $'s and the previous points to get
	\begin{eqnarray}
		\widehat{f}^\varphi
		Q^\varphi_j
		v\left(x_1,x_2\right)
		Q^\varphi_k
		&=&\sum_{l\in\mathbb Z}
		f\left(l\right) P_l^\varphi
		Q^\varphi_j
		v\left(x_1,x_2\right)
		Q^\varphi_k
	=\sum_{l\in\mathbb Z}
		f\left(l\right)
		Q^\varphi_j
		v\left(x_1,x_2\right)
		Q^\varphi_k
		P_{N-2, l-j}^\varphi
		\\
	&=&\sum_{l\in\mathbb Z}
		f\left(l\right)
		Q^\varphi_j
		v\left(x_1,x_2\right)
		Q^\varphi_k
		P_{l-j+k}^\varphi=
Q^\varphi_j
		v\left(x_1,x_2\right)
		Q^\varphi_k
		\widehat{f}^\varphi_{j-k}.\nonumber
	\end{eqnarray}
\end{proof}

\begin{pro}\label{pro:op_norm_estim}	Let $ p_j^\varphi $ and $ q_j^\varphi $ be defined as in \eqref{eq: def pj qj}. Let $ a, a^\prime\in\left[1,+\infty\right] $ such that $1/a + 1/a^\prime = 1$ and let $ u \in L^{2a'}\left(\mathbb R^3\right) $, $ \varphi\in L^{2a}\left(\mathbb R^3\right) $. It holds true that 
		\begin{equation}\label{itm:onep}
			\left\|u\left(x_1-x_2\right) p_1^\varphi\right\|
			\le\left\|\varphi\right\|_{L^{2a}\left(\mathbb R^3\right)}
			\left\|u\right\|_{L^{2a'}\left(\mathbb R^3\right)},
		\end{equation}
		\begin{equation}\label{itm:twops}
			\left\|
				p_1^\varphi
				u\left(x_1-x_2\right)
				p_1^\varphi
			\right\|
			\le\left\|
				\varphi
			\right\|_{L^{2a}\left(\mathbb R^3\right)}^2
			\left\|
				u
			\right\|_{L^{a'}\left(\mathbb R^3\right)}.
		\end{equation}
Moreover, if $ a\in\left[1,2\right] $, one has
		\begin{equation}\label{itm:fourps}
			\left\|u\left(x_1-x_2\right)p_1^\varphi p_2^\varphi\right\|
			\le\left\|\varphi\right\|_{L^{2a}\left(\mathbb R^3\right)}^2
			\left\|u\right\|_{L^{a'}\left(\mathbb R^3\right)}.
		\end{equation}
	\end{pro}
		
	\begin{proof}
		The estimate \eqref{itm:onep} directly follows from H\"older inequality,  
		\begin{align}
			\left\|
				u\left(x_1-x_2\right)
				p_1^\varphi
			\right\|^2
			&=\sup_{\left\|\psi\right\|=1}
			\left\|
				u\left(x_1-x_2\right)
				p_1^\varphi
				\psi
			\right\|^2
	\leq\sup_{\left\|\psi\right\|=1}
			\left\|\psi\right\|
			\sup_{x_2 \in\mathbb R^3}
			\langle
				\varphi,
				u^2\left(\cdot-x_2\right)
				\varphi
			\rangle
			\le
			\left\|\varphi\right\|_{L^{2a}\left(\mathbb R^3\right)}^2
			\left\|u\right\|_{L^{2a'}\left(\mathbb R^3\right)}^2.
		\end{align}
To prove \eqref{itm:twops}, we use \eqref{itm:onep}. Indeed,
		\begin{equation}
			\left\|
				p_1^\varphi
				u\left(x_1-x_2\right)
				p_1^\varphi
			\right\|
			\le\left\|
			\sqrt{\left|u\left(x_1-x_2\right)\right|}
				p_1^\varphi
			\right\|^2
			\le\left\|\varphi\right\|_{L^{2a}\left(\mathbb R^3\right)}^2
			\left\|\sqrt u\right\|_{L^{2a'}\left(\mathbb R^3\right)}^2
\le\left\|\varphi\right\|_{L^{2a}\left(\mathbb R^3\right)}^2
			\left\|u\right\|_{L^{a'}\left(\mathbb R^3\right)}.
		\end{equation}
		Finally, the estimate \eqref{itm:fourps} is a consequence of the Young inequality for the convolution (see \cite[Theorem 4.2]{LL}). More precisely, we have 
		\begin{align}
			\left\|
				u\left(x_1-x_2\right)
				p_1p_2
			\right\|^2
			&=\sup_{\left\|\psi\right\|=1}
			\left\langle
				\psi,
				\left(|\varphi\rangle\langle\varphi|\right)_1
				\left(|\varphi\rangle\langle\varphi|\right)_2
				u^2\left(x_1-x_2\right)
				\left(|\varphi\rangle\langle\varphi|\right)_1
				\left(|\varphi\rangle\langle\varphi|\right)_2
				\psi
			\right\rangle\nonumber
		\\
			&\le
			\int_{\mathbb R^6}dx_1dx_2\
			u^2\left(x_1-x_2\right)
			\left|\varphi\left(x_1\right)\right|^2
			\left|\varphi\left(x_2\right)\right|^2
			=\left\langle
				\left|\varphi\right|^2,
				u^2*\left|\varphi\right|^2
			\right\rangle\nonumber
		\\
			&\le\left\|
				\varphi
			\right\|_{L^{2p}\left(\mathbb R^3\right)}^2
			\left\|
				\varphi
			\right\|_{L^{2q}\left(\mathbb R^3\right)}^2
			\left\|
				u
			\right\|_{L^{2r}\left(\mathbb R^3\right)}^2.
		\end{align}
	
	\end{proof}

\end{document}